\UseRawInputEncoding
\documentclass{IEEEtaes}
\pdfoutput=1
\usepackage{amsmath,amsfonts}
\usepackage{algorithmic}
\usepackage{algorithm}
\usepackage{lineno,hyperref}
\usepackage{multirow}
\usepackage{epstopdf}
\usepackage{graphicx}
\usepackage{float}
\usepackage{caption}

\usepackage{hyperref}

\usepackage{color,array,amsthm}

\usepackage{graphicx}

\jvol{XX}
\jnum{XX}
\jmonth{XXXXX}
\paper{1234567}
\pubyear{2020}
\doiinfo{TAES.2020.Doi Number}

\newtheorem{theorem}{Theorem}
\newtheorem{lemma}{Lemma}
\newtheorem{definition}{Definition}
\newtheorem{remark}{Remark}
\newtheorem{assumption}{Assumption}
\usepackage[inkscapelatex=false]{svg}
\usepackage[numbers,sort&compress]{natbib}

\modulolinenumbers[5]

\begin{document}

\title{Singularity-Avoidance Prescribed Performance Attitude Tracking of Spacecraft}

\author{Jiakun Lei}
\affil{School of Aeronautics and Astronautics, Zhejiang University, Hangzhou 310027, China} 

\author{Tao Meng}
\affil{School of Aeronautics and Astronautics, Zhejiang University, Hangzhou 310027, China} 

\author{Weijia Wang}
\affil{School of Aeronautics and Astronautics, Zhejiang University, Hangzhou 310027, China}

\author{Heng Li}
\affil{School of Aeronautics and Astronautics, Zhejiang University, Hangzhou 310027, China}

\author{Zhonghe jin}
\affil{School of Aeronautics and Astronautics, Zhejiang University, Hangzhou 310027, China}

%% \author{FOURTH D. AUTHOR}
%% \affil{University of Colorado, Colorado, USA}

%\receiveddate{This paragraph of the first footnote will contain the date on which you submitted your paper for review, which is populated by IEEE. It is IEEE style to display support information, including sponsor and financial support acknowledgment, here and not in an acknowledgment section at the end of the article. For example, ``This work was supported in part by the U.S. Department of Commerce under Grant BS123456.'' }
%% \accepteddate{XXXXX XX XXXX}
%% \publisheddate{XXXXX XX XXXX}

%\corresp{The name of the corresponding author appears after the financial information, e.g. {\itshape (Corresponding author: M. Smith)}. Here you may also indicate if authors contributed equally or if there are co-first authors.}

\authoraddress{Author's address: School of Aeronautics and Astronautics, Zhejiang University, Hangzhou 310027, China;E-mail:12124010@zju.edu.cn,Zhejiang Key Laboratory of Micro-nano satellite Research, Hangzhou 310027, China}

%\editor{Mentions of supplemental materials and animal/human rights statements can be included here.}
%\supplementary{Color versions of one or more of the figures in this article are available online at {http://ieeexplore.ieee.org}.}

\markboth{Lei ET AL.}{SAPPC of Attitude Tracking}
\maketitle

\begin{abstract}	The attitude tracking problem with preassigned performance requirements has earned tremendous interest in recent years, and the Prescribed Performance Control (PPC) scheme is often adopted to tackle this problem. Nevertheless, traditional PPC schemes have inherent problems, which the solution still lacks, such as the singularity problem when the state constraint is violated and the potential over-control problem when the state trajectory approaches the constraint boundary.
	This paper proposes a Singularity-Avoidance Prescribed Performance Control scheme (SAPPC) to deal with these problems. A novel shear mapping-based error transformation is proposed to provide a globally non-singular error transformation procedure, while a time-varying constraint boundary is employed to exert appropriate constraint strength at different control stages, alleviating the potential instability caused by the over-control problem. Besides, a novel piece-wise reference performance function (RPF) is constructed to provide a relevant reference trajectory for the state responding signals, allowing precise control of the system's responding behavior. Based on the proposed SAPPC scheme, a backstepping controller is developed, with the predefined-time stability technique and the dynamic surface control technique employed to enhance the controller's robustness and performance. Finally, theoretical analysis and numerical simulation results are presented to validate the proposed control scheme's effectiveness and robustness.
\end{abstract}


\section{INTRODUCTION}
I{\scshape n} 
 In recent years, contemporary space missions have appeared with higher complexity and task requirements, and some advanced performance requirements are required to be satisfied. Since the traditional controller is hard to achieve these given performance requirements prior, controllers with the guarantee of prescribed performance have recently been of high research interest. Driven by this requisition, this paper focuses on developing a relevant controller for the attitude tracking task, ensuring that all the preassigned performance requirements can be satisfied even when a big external disturbance exists. Further, this article's another key point is the accurate control of the system responding behavior, such as transient and steady-state behavior.

For the fundamental attitude control problem under disturbances, many efforts have been devoted to this issue, and a large number of techniques and control schemes are illustrated in the existing literature, such as the sliding mode control \cite{zhang_synchronization_2019,zhuang_robust_2021,lu_sliding_2012}, model predictive control \cite{mammarella_attitude_2019,chai_dual-loop_2022,golzari_quaternion_2020} etc, parameter adaptive strategy is often combined with these control schemes to provide an active disturbance compensation for the system\cite{thakur_adaptive_2015}.  However, the majority of these proposed controllers can only guarantee the asymptotical stabilization of the system, which implies that the system will cost infinite time to settle down theoretically.
Driven by the requisition for an explicit expression of the settling time, predefined-time stability is established by Sanchez-Torres in \cite{sanchez-torres_discontinuous_2014,sanchez-torres_predefined-time_2015}, providing a direct design of the upper boundary of the settling time. The predefined-time stability theory has been applied to various control schemes to provide a higher convergence rate for the system, especially in the space mission scenarios \cite{wang_continuous_2022, wang_attitude_2020, shi_satellite_2021}. 
However, although applying the predefined time stability theory enhances the controller's performance, the control effect is still unable to guarantee prior.

In order to achieve higher constraint ability on the state variable and the preassigned performance requirements, Bechlioulis and Rovithakis developed a unique methodology called Prescribed Performance Control(PPC). As they stated in \cite{bechlioulis_adaptive_2009,bechlioulis_robust_2008}, Due to these specific performance requirements, the original system is a constrained system. However, applying a homeomorphic error transformation function can transform the original constrained system into an equivalently unconstrained one. It has been strictly proved that if the convergence of the translated system can be satisfied, the original performance requirements will be able to achieve simultaneously.
Following the guidance of such an idea, the PPC scheme has been applied to spacecraft control problems to show its effectiveness, as listed in \cite{zhang_simple_2019,zhuang_fixed-time_2021,gao_finite-time_2021,zhang_prescribed_2019,shi_actor-critic-based_2021,huang_fault-tolerant_2020,luo_low-complexity_2018,xia_neuroadaptive_2021,wei_learning-based_2019,WANG2022}.

Although existing literature related to the PPC issue has proposed many effective schemes, there are still some problems that are worth noticing. Firstly the singularity problem. The traditional PPC scheme only works properly when the state trajectory stays in the constraint region. However,  although it is possible to technically guarantee that the state trajectory stays in the constraint region at $t = 0$, it cannot be ensured that the state trajectory will stay in the constraint region for the whole control process, especially when there exists big disturbance. Secondly, the potential instability problem caused by over-control. For the traditional PPC scheme, there exists a contradiction in the constraint strength and the control effect. If the constraint's strength is tight, the state trajectory will be affected by strong repulsion exerted by the constraint boundary, leading to the chattering of the system; if the constraint is relatively loose, the constraint will not be strong enough to restrain the state trajectory, and this will lead to a decreasing in the control effect. This problem limited the performance of the traditional PPC scheme, and the best trade-off between the stability and the performance is hard to evaluate. Thirdly, the existing PPC scheme provides a constraint region for the state trajectory. However, there is still a lack of an efficient way to precisely control the system's transient behavior. We found that the solution to these problems is still an open problem, which is worth further investigation.

In order to tackle these stated problems, a novel PPC scheme named singularity-avoidance prescribed performance control (SAPPC) is proposed in this paper. The proposed SAPPC scheme's main structure consists of three parts: a novel shear mapping-based error transformation function (abbr. SMETF), a time-varying state constraint, and a novel reference trajectory function (abbr. RPF). Firstly, by employing the shear mapping, the proposed SMETF realizes a globally non-singular error transformation, ensuring that the system will not trap into singularity when the state trajectory is out of the constraint region. This characteristic ensures that the state trajectory is able to converge back to the constraint region even when the system trajectory is pulled out of the constraint region by external disturbances.
Secondly, the time-varying constraint boundary is proposed to help alleviate the chattering caused by the over-control problem. The time-varied state constraint related to the RPF will provide appropriate constraint strength at different stages, alleviating the over-control problem. Further, the novel piece-wise RPF provides a relevant reference for the state responding, allowing the precise design of the system's behavior.

The main contribution of this paper is concluded as follows:

$1.$ The shear space mapping is first introduced to the conventional PPC scheme to provide a global non-singular error transformation procedure.

$2.$ A time-varying state constraint boundary is constructed to alleviate the potential instability problem at the steady-state caused by the over-control. The proposed state constraint boundary will change according to the reference function, providing a strong restrain ability at the convergence stage and a loose one at the steady-state. 

$3.$ A novel reference performance function is built to provide a precise design of the system settling time and steady-state error.

This paper is organized as follows: In Section \ref{preliminary}, notations in this paper and some useful lemmas are introduced for the following analysis. Problem formulation is accomplished in Section \ref{model} with the introduction of the system model. The detailed elaboration of the proposed SAPPC scheme is presented in Section \ref{solution}, and a relevant controller is built based on the proposed SAPPC scheme. Numerical simulation results, including regular case analysis, comparison, and Monte Carlo simulation, are presented in Section \ref{simulation}.

	\section{Preliminaries}
\label{preliminary}
\subsection{Notations}
For further analysis, the following notations are defined in this paper.
$\boldsymbol{I}_n$ denotes the $n \times n$ identity matrix, $\|\cdot\|$ denotes the Euclidean norm of a given vector or the induced norm of a given matrix. $\mathfrak{R}_i$ represents the Earth-Central-Inertial frame, while $\mathfrak{R}_b$ denotes the spacecraft body-fixed frame. 
Additionally,
for any vector $\boldsymbol{b}\in \mathbb{R}^{3}$, the operator $\boldsymbol{b}^{\times}$ denotes the $3 \times 3$ skew-symmetric matrix for cross manipulation, i.e. $\boldsymbol{b}^{\times}\boldsymbol{s} = \boldsymbol{b}\times\boldsymbol{s}$.
$\text{diag}\left(b_i\right)$ represents a diagonal matrix whose diagonal line is consisted by the components of the given vector $\boldsymbol{b}$, such that:
\begin{equation}
	\text{diag}\left(b_{i}\right) = 
	\begin{bmatrix}
		b_{1} & 0 & 0 \\
		0 & b_{2} & 0 \\
		0 & 0 & b_{3}
	\end{bmatrix}
\end{equation} 
where $b_{i}$ represents the $i$ th component of the given vector. Similarly, $\text{vec}\left(b_{i}\right)$ denotes a column vector whose components are $b_{1},b_{2},...b_{i}$ correspondingly, i.e. $\text{vec}\left(b_{i}\right) = \left[b_{1},b_{2},...b_{i}\right]^{\text{T}}$.

\subsection{Mathematical Lemma}
\begin{lemma}
	\label{lemma_tanh}
	For any $\mu > 0$ and $x\in\mathbb{R}$ , the following formula will be satisfied \cite{walls_globally_2005}:
	\begin{equation}
		0 \le |x| - x\tanh\left(\frac{x}{\mu}\right) \le 0.2785\mu 
	\end{equation}
\end{lemma}

\begin{lemma}
	\label{lemma_PTS}
	\cite{sanchez_class_2018} Assume that there exists a continuous positive-definite radially unbounded function denoted as $V:\mathbb{R}^{n\times n} \to \mathbb{R}_{+} \cup\left\{0\right\}$, such that:
	\begin{equation}
		\begin{aligned}
			V\left(0\right) &= 0 \\ \quad V\left(\boldsymbol{x}\right) > 0, & \forall \boldsymbol{x} \neq \boldsymbol{0}
		\end{aligned}
	\end{equation} For arbitrary real number $T_{c} \in \left(0,+\infty\right)$, $a \in \left(0,+\infty\right)$ and $p \in$ $\left(0,1\right)$. Take the time-derivative of $V$, if the time derivative of $V$ satisfies:
	\begin{equation}
		\dot{V} \le -\frac{1}{pT_{c}}e^{aV^{p}}V^{1-p}
	\end{equation}
	then, the origin of the system is globally predefined-time stabled, and the settling time will satisfy $T_{\text{set}} \le T_{\text{max}} = \frac{1}{a}T_{c}$.
\end{lemma}

\begin{lemma}
	\label{lemma_chebyshev}
	(Chebyshev's sum inequality) \cite{hardy_1952_inequalities} For $n\in \mathbb{N}^{+}$, if $a = \left\{a_{1},...a_{n}\right\}$ and $b = \left\{b_{1},...b_{n},\right\}$ are two similarly ordered real number sequences such that $a_{1} \le a_{2} \le...a_{n}$ and $b_{1} \le b_{2} \le,...b_{n}$ or  $a_{1} \ge a_{2} \ge...a_{n}$ and $b_{1} \ge b_{2} \ge,...b_{n}$, the following inequalities will be always holds:
	\begin{equation}
		\sum_{i=1}^{n}a_{i}b_{i} \ge \frac{1}{n}\left(\sum_{i=1}^{n}a_{i}\right)\left(\sum_{i=1}^{n}b_{i}\right)
	\end{equation}
\end{lemma}

\begin{lemma}
	\label{lemma_mean}
	(Inequality of arithmetic and geometric means)\cite{hardy_1952_inequalities} 
	Consider a positive number sequence as $a = \left\{a_{1},a_{2},...a_{n}\right\}\left(n\in\mathbb{N}^{+}\right)$, the following property will be obtained:
	\begin{equation}
		\left(\prod_{i=1}^{n}a_{i}\right)^{\frac{1}{n}} \le
		\frac{1}{n}\sum_{i=1}^{n}a_{i}
	\end{equation}
\end{lemma}

\begin{lemma}
	\label{lemma_beta}
	 For a series of  positive numbers $x_{i}\left(i = 1,2...n\right)$ and a constant number $\beta \in \left(0,1\right)$, one can be obtained that \cite{chen_neural_2022}:
	\begin{equation}
		\sum_{i=1}^{n}x_{i}^{\beta} \ge \left(\sum_{i=1}^{n}x_{i}\right)^{\beta}
	\end{equation}
\end{lemma}

\begin{lemma}
	\label{lemma_last}
	Considering the function $c\left(x\right)$ expressed as $c\left(x\right) = e^{x^{p}}x^{1-p} - x$, for $x\in \left(0,+\infty\right)$,  $p\in\left(0,1\right)$, the following result will be satisfied
	\begin{equation}
		c\left(x\right) = e^{x^{p}}x^{1-p}-x \ge \left(x^{p}+1\right)x^{1-p}-x=x^{1-p}>0
	\end{equation}
	
\end{lemma}
\begin{theorem}
	
	\label{theoremone}
	Assume that there exists a continuous positive-definite radially unbounded function denoted as $V:\mathbb{R}^{n\times n} \to \mathbb{R}_{+} \cup\left\{0\right\}$, such that:
	\begin{equation}
		\begin{aligned}
			V\left(0\right) &= 0 \\ \quad V\left(\boldsymbol{x}\right) > 0, & \forall \boldsymbol{x} \neq \boldsymbol{0}
		\end{aligned}
	\end{equation} For arbitrary real number $T_{c} \in \left(0,+\infty\right)$, $a \in \left(0,+\infty\right)$ and $p \in$ $\left(0,1\right)$. Take the time-derivative of $V$, if the time derivative of $V$ satisfies:
	\begin{equation}
		\label{dVsystem}
		\dot{V} \le -\frac{1}{pT_{c}}e^{aV^{p}}V^{1-p} + \Theta
	\end{equation}
	where $\Theta \in \left(0,+\infty\right)$ is a residual term. therefore, the dynamical system will be practically predefined-time stabled, and the solution of the dynamical system (\ref{dVsystem}) will converge to a residual set in a predefined-time $T_{\text{set}} \le \frac{1}{a\mu}T_{c}$, while $\mu$ is a positive constant satisfies $\mu \in\left(0,1\right)$. The residual set of the solution can be expressed as follows:
	\begin{equation}
		\left\{\boldsymbol{x}|e^{aV^{p}\left(\boldsymbol{x}\right)}V^{1-p}\left(\boldsymbol{x}\right) \le \frac{pT_{c}\Theta}{1-\mu}\right\}
	\end{equation}
\end{theorem}
\begin{proof}
	
	Consider the dynamical system expressed in equation (\ref{dVsystem}), by applying the positive constant $\mu$, the inequality (\ref{dVsystem}) can be rearranged into the following form, expressed as follows:
	\begin{equation}
		\dot{V} \le -\frac{\mu}{pT_{c}}e^{aV^{p}}V^{1-p}  
		-\frac{1-\mu}{pT_{c}}e^{aV^{p}}V^{1-p} + \Theta  
	\end{equation}
	therefore, for $\Theta \le \frac{1-\mu}{pT_{c}}\exp\left(aV^{p}\right)V^{1-p}$, it should be noted that $ \dot{V} \le -\frac{\mu}{pT_{c}}e^{aV^{p}}V^{1-p}$ will be satisfied. Further, consider the strictly equality condition as follows:
	\begin{equation}
		\label{dVV}
		\dot{V} = -\frac{\mu}{pT_{c}}e^{aV^{p}}V^{1-p}
	\end{equation}
	Integrate the equation, the equation (\ref{dVV}) can be rearranged into the following form as
	\begin{equation}
		\frac{\mu}{pT_{c}}T\left(\boldsymbol{x}_{0}\right) = 
		\int_{V_{0}}^{0}\left(-e^{-aV^{p}}V^{p-1}\right)dV
	\end{equation}
	thus, we have:
	\begin{equation}
		\frac{\mu}{pT_{c}}T\left(\boldsymbol{x}_{0}\right) = 
		\frac{1}{ap}e^{-aV^{p}}\big|^{0}_{V_{0}} = \frac{1}{ap}\left[1 - e^{-aV_{0}^{p}}\right] \le \frac{1}{ap}
	\end{equation}
	Finally, we have the following conclusion that the system will fall into the residual region in a predefined-time related to $a$, $\mu$, $T_{c}$ as
	\begin{equation}
		\sup_{\boldsymbol{x}_{0}\in\mathbb{R}^{n}}T\left(\boldsymbol{x}_{0}\right) = \frac{1}{a\mu}T_{c}
	\end{equation}
	and the residual of the system is  expressed as follows:
	\begin{equation}
		\label{Rset}
		\left\{\boldsymbol{x}|e^{aV^{p}\left(\boldsymbol{x}\right)}V^{1-p}\left(\boldsymbol{x}\right) \le \frac{pT_{c}\Theta}{1-\mu}\right\}
	\end{equation}
	it should be noted that for $\Theta = 0$, the residual set $\ref{Rset}$ of the system will be equal to $0$. Thus, the conclusion of Theorem 1 is the same as the lemma {\ref{lemma_PTS}}. The conclusion of theorem 1 can be regarded as a generalized condition of lemma {\ref{lemma_PTS}}, as real systems may suffer from various kinds of perturbations. 
\end{proof}

	\section{Problem Formulation}
\label{model}
\subsection{Attitude Kinematics and Dynamics}
%姿态动力学系统%
Consider the attitude error kinematic and dynamic equation of a rigid-body spacecraft expressed in the normalized quaternion, the attitude error system can be modeled as follows:
\cite{wertz_2012_spacecraft}
\begin{equation}
	\centering
	\label{errorsystem}
	\begin{aligned}
		\dot{\boldsymbol{q}}_{ev} &= \boldsymbol{\varGamma}\left(\boldsymbol{q}_e\right)\boldsymbol{\omega}_e\\
		\dot{q_{e0}} &= -\frac{1}{2}\boldsymbol{q}^{\text{T}}_{ev}\boldsymbol{\omega_{e}} \\
		\boldsymbol{J}\dot{\boldsymbol{\omega}}_e =  \boldsymbol{J}\boldsymbol{\omega}^{\times}_e\boldsymbol{C}_e\boldsymbol{\omega}_d 
		&- \boldsymbol{J}\boldsymbol{C}_e\dot{\boldsymbol{\omega}}_d
		-\boldsymbol{\omega}_s^{\times}\boldsymbol{J}\boldsymbol{\omega}_s + \boldsymbol{u} + \boldsymbol{d}
	\end{aligned}
\end{equation}
where $\boldsymbol{q}_{e} = \left[q_{e0}\quad \boldsymbol{q}_{ev}^{\text{T}}\right]^\text{T}\in\mathbb{R}^{4}$ denotes the attitude error quaternion of the spacecraft's body-fixed axes with respect to the inertial frame, $q_{e0}\in\mathbb{R}$ and $\boldsymbol{q}_{ev} = \left[q_{ev1},q_{ev2},q_{ev3}\right]^{\text{T}}\in\mathbb{R}^3$ represents the scalar part and the vector part of the attitude error quaternion, respectively. $\boldsymbol{\omega}_{e}= \left[\omega_{e1},\omega_{e2},\omega_{e3}\right]^{\text{T}} \in\mathbb{R}^{3}$ represents the error angular velocity of the spacecraft with respect to the inertial frame $\mathfrak{R}_i$ expressed in the current body-fixed frame $\mathfrak{R}_b$. The  inertial matrix of the spacecraft with respect to the body-fixed frame $\mathfrak{R}_b$ is denoted as $\boldsymbol{J}\in\mathbb{R}^{3\times 3}$, which is a known constant matrix. $\boldsymbol{u} = [u_{1},u_{2},u_{3}]^{\text{T}}\in\mathbb{R}^3$ denotes the control input, while $\boldsymbol{d}\in\mathbb{R}^3$ denotes the unknown lumped external disturbances. $\boldsymbol{C}_e \in \mathbb{R}^{3\times 3}$ represents the attitude transformation matrix calculated by $	\boldsymbol{C}_e = \left(q^2_{e0}-\boldsymbol{q}_{ev}^\text{T}\boldsymbol{q}_{ev}\right)\boldsymbol{I}_3+2\boldsymbol{q}_{ev}\boldsymbol{q}_{ev}^\text{T}-2q_{e0}\boldsymbol{q}^{\times}_{ev}$.

$\boldsymbol{\varGamma}\left(\boldsymbol{q}_{e}\right)\in\mathbb{R}^{3 \times 3}$ represents the Jacobian matrix of the attitude error quaternion $\boldsymbol{q}_{e}$ that can be expressed as follows:
\begin{equation}
	\label{Jacobian}
	\boldsymbol{\varGamma}\left(\boldsymbol{q}_{e}\right) = \frac{1}{2}\left(q_{e0}\boldsymbol{I}_3 + \boldsymbol{q}_{ev}^{\times}\right)
\end{equation}
 For further analysis, the normalized desired attitude quaternion is denoted as $\boldsymbol{q}_d = \left[q_{d0} \quad \boldsymbol{q}_{dv}^{\text{T}}\right]^\text{T}\in\mathbb{R}^4$, and the desired angular velocity of the spacecraft with respect to the inertial frame $\mathfrak{R}_i$ expressed in the target body-fixed frame $\mathfrak{R}_{d}$ is denoted as $\boldsymbol{\omega}_d\in \mathbb{R}^{3}$.
According to \cite{shi_satellite_2021,wei_learning-based_2019}, we make following assumptions for the synthesis of the proposed control scheme.

\begin{assumption}
	{\label{assump_J}}
	The inertial matrix $\boldsymbol{J}$ of the spacecraft is a known symmetric positive-definite matrix. We define the maxima and the minima of its eigen value as $J_{\text{min}}$ and $J_{\text{max}}$ respectively, thus we have:
	\begin{equation}
		J_{\text{min}}\boldsymbol{x}^{\text{T}}\boldsymbol{x} \le \boldsymbol{x}^{\text{T}}\boldsymbol{J}\boldsymbol{x} \le
		J_{\text{max}}\boldsymbol{x}^{\text{T}}\boldsymbol{x}
	\end{equation}
\end{assumption} 

\begin{assumption}
	{\label{assump_Dis}}
	The external disturbance is unknown but bounded by a known constant, i.e., $\|d\|$ $\le D_{m}$.
\end{assumption}

\begin{assumption}
	\label{assump_Qd}
	The desired attitude trajectory $\boldsymbol{q}_{d}$ along with its derivatives $\dot{\boldsymbol{q}}_{d}$ and $\ddot{\boldsymbol{q}}_{d}$ are all smooth functions and are bounded by positive constants. Therefore, there exists a positive constant $U_{0}$ such that $\boldsymbol{q}_{d}^{2} + \dot{\boldsymbol{q}}_{d}^{2} + \ddot{\boldsymbol{q}}_{d}^{2} \le U_{0}$ is always holds.
\end{assumption}

\subsection{Control Objective}
The controller is expected to guarantee the satisfaction of all the performance requirements even when significant disturbances exist. Further, the state trajectory is expected to track a given curve, enabling us to assign the system's transient behavior directly.

\section{Problem Solution}
\label{solution}
The detail elaboration of the proposed SAPPC scheme is organized as follows: A brief introduction of the proposed SAPPC structure is elaborated in section IV.\ref{Structure}, while detailed introduction of the novel RPF, time-varying state constraint and the novel shear mapping-based error transformation function (SMETF) is presented in section IV.\ref{PF}, IV.\ref{Constraint} and IV.\ref{SMETF} respectively.
Further, a thorough theoretical comparison analysis of the traditional PPC scheme and the proposed SAPPC scheme is presented in section IV.\ref{Comparison}. Finally, a backstepping controller wtih predefined-time stability and dynamic surface filter is developed based on the SAPPC scheme, detailed in section IV.\ref{ControlLaw}.

\subsection{Structure of the proposed SAPPC scheme}
\label{Structure}	

Compared with the traditional PPC scheme, some significant changes are applied in our schemes, and the main change is embodied in the error transformation procedure. The shear mapping-based error transformation function (SMETF) can be regarded as a combination of the original homeomorphic mapping function and the shear mapping, which can perform the error transformation without singularity. Further, the constant constraint boundary is replaced by the time-varying constraint boundary to exert appropriate constraint strength at different control stages. Besides, the original performance function is replaced by a newly-designed reference trajectory function (RPF). A comparison sketch map is illustrated in the Figure [\ref{SAPPCstruct}].
·
\begin{figure}[hbt!]
	\centering 
	\includegraphics[scale=0.25]{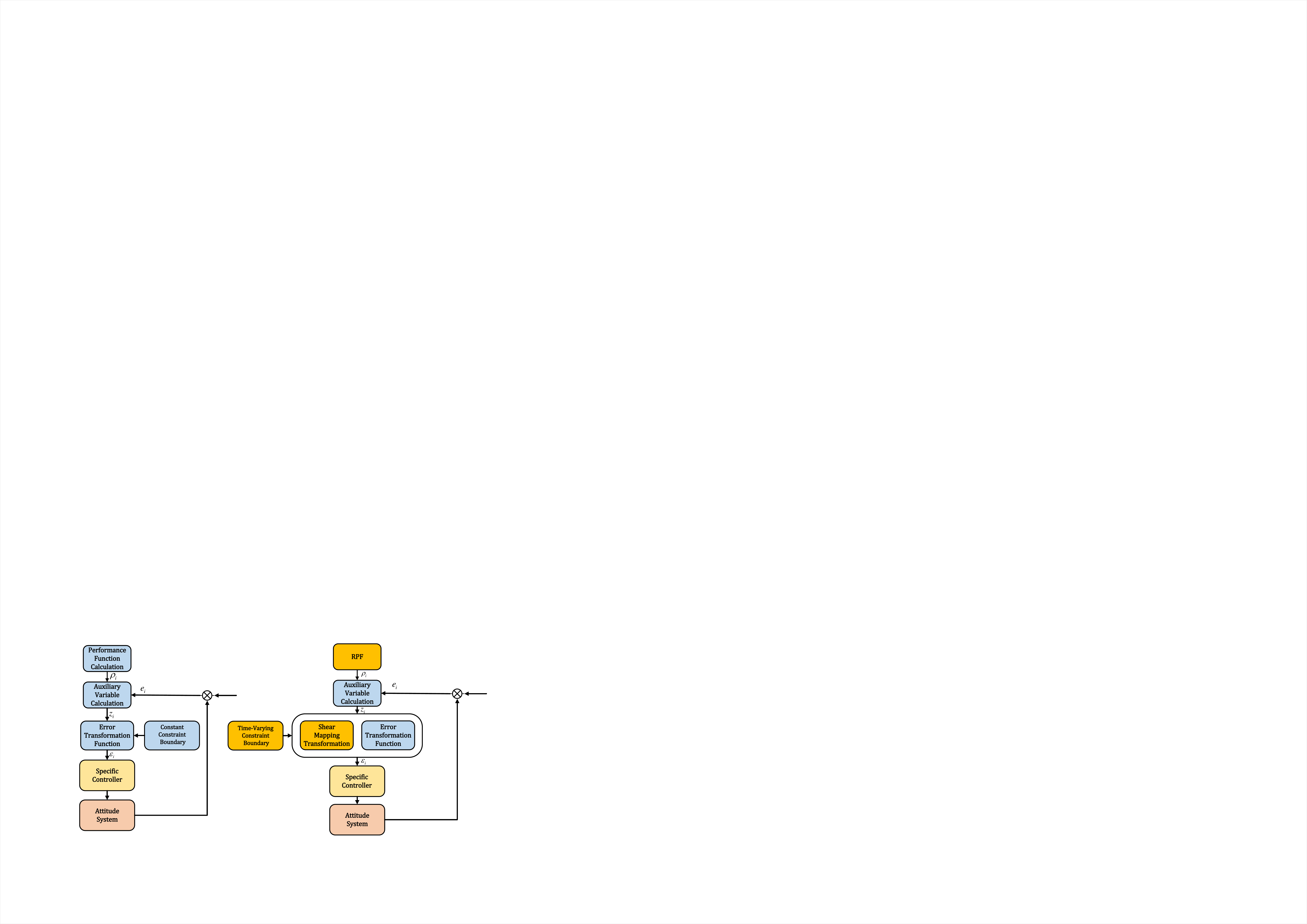}  
	\caption{Structure Comparison of the Traditional PPC and the SAPPC}
	\label{SAPPCstruct}
\end{figure}

	\subsection{The Reference Performance Function(RPF)}
\label{PF}
 Inspired by the performance function applied in traditional PPC schemes, if the actual state responding trajectory could track a given reference trajectory tightly, we would be able to control the system's transient and steady-state behavior precisely.
This section constructs a smooth piece-wise function named Reference Performance Function (RPF) to directly assign some significant performance index, such as the system settling time and terminal control error. The proposed RPF is expressed as follows:

\begin{equation}
	\label{RPF}
	\rho(t) = \begin{cases}
		\rho_{e}(t)=\left(\rho_{e0}-\rho_{e\infty}\right)e^{-lt}+\rho_{e\infty}&  0 \le t<t_{1} \\ 
		\rho_{p}(t)={a_1t^2+a_2t+a_3}                                          &  t_1\le t<t_2\\
		\rho_{c}(t)=g_\infty                                                   &  t_2 \le t
	\end{cases}
\end{equation}

In the expression (\ref{RPF}), $t_{2}\in\left(0,+\infty\right)$ is the settling time of the RPF that should be indicated directly, $\rho_{e0}$, ${\rho_{e\infty}}$ are the initial value and the asymptotes's value of the exponential function part respectively, while $g_\infty $ denotes the terminal value of the RPF, i.e. $\rho_{c}\left(t\right) = g_{\infty}$. $a_1$, $a_2$, $a_3$ and $t_1$ are the coefficients used to decide the RPF, which needs solving later. $t_{1}$ is the time instant such that $\rho_{e}\left(t_{1}\right) = \rho_{p}\left(t_{1}\right)$.
The expression of the whole RPF can be obtained by solving these following equations.
\begin{equation}
	\label{RPFsolve}
	\begin{matrix}
		\left[\dfrac{k}{2} \left(t_2 - t_1\right)-1\right] \left(\rho_{e0}-\rho_{e\infty}\right)e^{-lt_1}-\rho_{e\infty}+g_\infty=0 \\
		a_1 = \left(\rho_{e0}-\rho_{e\infty}\right)e^{-lt_1}/2\left(t_1-t_2\right)\\
		a_2 = -2a_1t_2\\
		a_3 = g_\infty+a_1t_2^2
	\end{matrix}
\end{equation}$  $
Take the time-derivative of each segmented part of the proposed RPF, the left time-derivative of the RPF is the same as the right at the segment time instant, i.e., $\dot{\rho}_{e}(t_1) = \dot{\rho}_{p}(t_1)$, $\dot{\rho}_{p}(t_2) = \dot{\rho}_{c}(t_2)$. Therefore, the given RPF is smooth and differentiable on $t\in\left(0,+\infty\right)$. Note that the selecting of the parameter should guarantee a real number solution exists for the equation (\ref{RPFsolve}).

 A reasonable principle to select a $\rho_{e0}$is let $\rho_{e0} = e\left(0\right)$, where $e\left(0\right)$ denotes the initial condition of the error state variable. Further, it should be guaranteed that  $\rho_{e\infty} < g_{\infty}$ is satisfied, or the solution will not exists. Here we given an example of the RPF curve, choose parameter as $t_{2} = 15s$, $\rho_{e0} = 0.25$, $\rho_{e\infty} = 1e-6$, $g_{\infty} =3e-5$, the corresponding RPF is illustrated in Figure [\ref{figrpf}].

\begin{figure}[hbt!]
	\centering 		\includegraphics[scale=0.3]{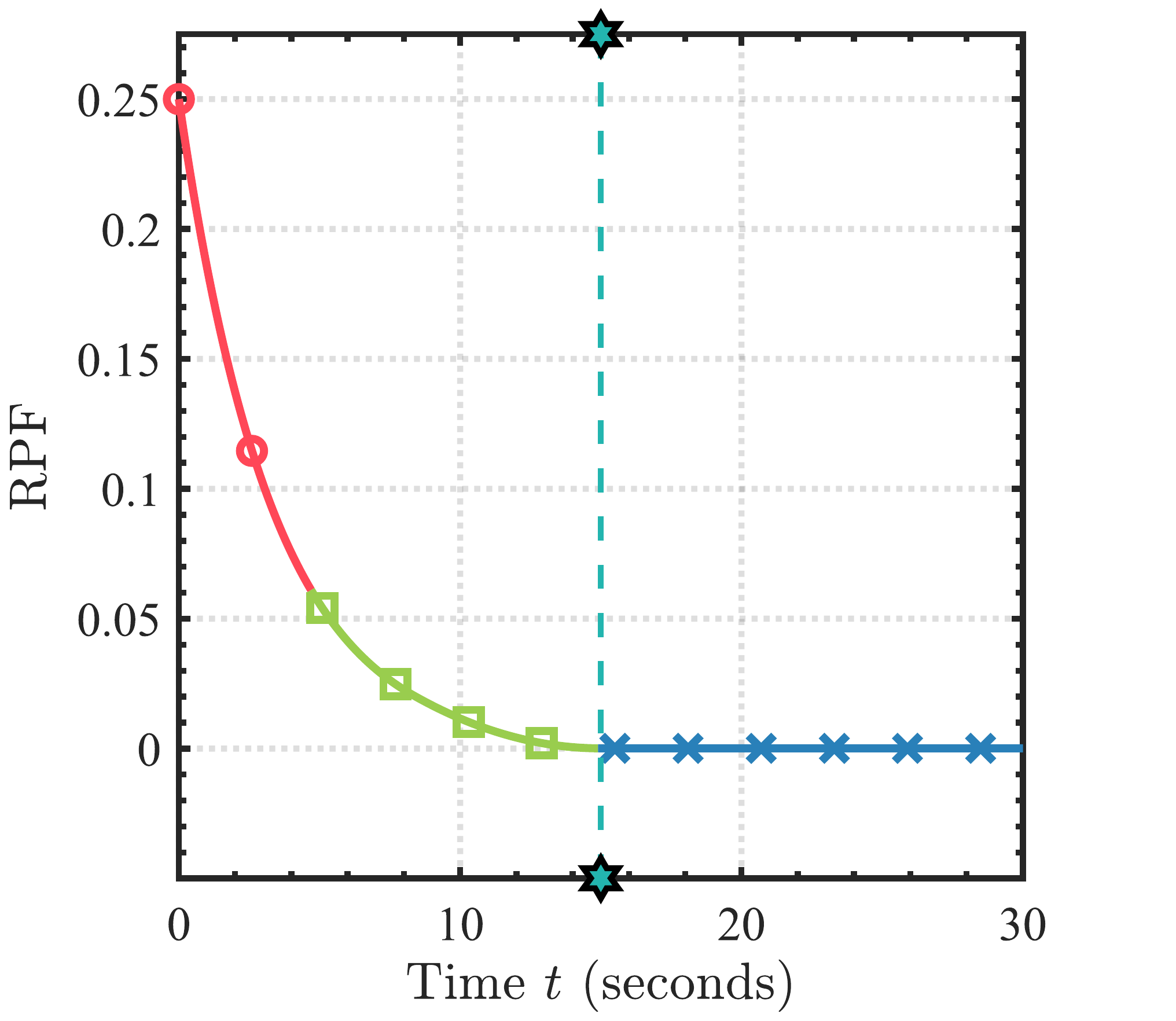}  
	\caption{The proposed Reference Trajectory Function(RPF)}  
	\label{figrpf}
\end{figure}

As Figure [\ref{figrpf}] shows, each part of the RPF is illustrated in different color separately, while the star marker represents the preassigned settling time of the RPF, indicated by $t_2 = 15s$.
The symbol of the RPF is decided by state variable's initial condition as follows:
\begin{equation}
	\begin{aligned}
		\rho\left(t\right) \le 0 \quad	\dot{\rho}\left(t\right) &\ge 0, \quad if e\left(0\right) \le 0 \\
		\rho\left(t\right) \ge 0 \quad	\dot{\rho}\left(t\right) &\le 0, \quad if e\left(0\right) \ge 0 \\
	\end{aligned}
\end{equation}

	\subsection{The Time-varying constraint boundary}
%Brief Introduction Of PPC problem%
\label{Constraint}

As mentioned in section IV.\ref{PF}, to accurately meet those performance requirements, the system state trajectory needs to track the RPF tightly. For further analysis, we define an auxiliary variable as $\boldsymbol{z}_{s}\left(t\right)= \left[z_{s1}\left(t\right),...z_{si}\left(t\right)\right]^{\text{T}}\left(i=1,2,3\right)$ and an RPF as $\boldsymbol{\rho}\left(t\right) = \left[\rho_{1}(t),...\rho_{i}(t)\right]^
{\text{T}}\left(i=1,2,3\right)$, accordingly, each component of the auxiliary variable $\boldsymbol{z}_{s}(t)$ can be calculated as follows:
\begin{equation}
	z_{si}\left(t\right) = \frac{e_{i}(t)}{\rho_{i}\left(t\right)}
\end{equation}
where $e_{i}\left(t\right)$ represents the $i$ th component of the original system error state variable $\boldsymbol{e}\left(t\right) = 
\left[e_{1}(t),...e_{i}(t)\right]^{\text{T}}\left(i=1,2,3\right)$.
Further, define a constraint boundary vector as $\boldsymbol{\delta}\left(t\right) = [\delta_{1}(t),...\delta_{i}(t)]^{\text{T}}\left(i=1,2,3\right)$, where $\delta_{i}\left(t\right)$ is defined as follows:
\begin{equation}
	\label{delta_i}
	\delta_i(t) =\frac{B_{0}}{|\rho_{i}(t)|}
\end{equation}
$B_{0} > 0$ is a positive constant that should be indicated.
The state constraint in our SAPPC scheme is expressed as follows:
\begin{equation}
	\label{constraint}
	1 - \delta_i(t)< z_{si}(t) < 1 + \delta_i(t), \quad   
\end{equation} 

 Considering the time-derivative of the constraint boundary $\delta_{i}\left(t\right)$,
 $\frac{d(|{\rho_{i}}|)}{dt} \le 0$ will be always satisfied. Take the derivative of $\delta_i(t)$ with respect to time, one can be obtained that for $t\in \left(0,t_{2}\right]$, $\dot{\delta}_{i}\left(t\right) > 0$ will be satisfied; for $t\in\left(t_{2},+\infty\right)$, $\dot{\delta}_{i}\left(t\right) = 0$ will be hold.
\begin{remark}
	It can be observed that the time-varying constraint will ensure that the tracking error to the RPF will converge to a preassigned region. Unlike those constant proportional constraints applied in the PPC schemes, this kind of constraint will alleviate the chattering caused by the over-control problem. The detailed elaboration about the chattering problem will be presented in section IV.\ref{Comparison}.
\end{remark}

	\subsection{the Novel Shear Mapping-based Error Transformation Function (SMETF)}
\label{SMETF}
This section proposes a novel shear mapping-based error transformation function (SMETF) to deal with the singularity problem. 
Considering any relevant typical homeomorphic error transformation function applied in the traditional PPC scheme, we denote it as $\mathcal{R}\left(\cdot\right)$.
For any given auxiliary variable $\boldsymbol{z}_{s}$, define its corresponding translated variable calculated by the traditional error transformation function as $\boldsymbol{\varepsilon}\left(t\right) = \left[\varepsilon_{1}(t),...\varepsilon_{i}(t)\right]^{\text{T}}\left(i=1,2,3\right)$, thus we have
\begin{equation}
	\varepsilon_{i} =  \mathcal{R}\left(z_{si}\right)
\end{equation}

\begin{definition}
	Linear mapping defined as follows is called the two-dimensional shear mapping \cite{clifford_1886_common}
	
	\begin{equation}
		\label{sm}
		\begin{bmatrix}
			x_{1} \\ y_{1}
		\end{bmatrix} = 
		\begin{bmatrix}
			1 & \tan\theta \\
			0 & 1
		\end{bmatrix}
		\begin{bmatrix}
			x_{0} \\ y_{0}
		\end{bmatrix}
	\end{equation}
	where  $\left(x_0,y_0\right)$ denotes any given point in the original space and $\left(x_1,y_1\right)$ denotes the corresponding point in the transformed space.  $\theta\in\left(0,\frac{\pi}{2}\right)$ is the shear angle of the shear mapping, which influences the inclination degree of the mapped space with respect to the original one.

\end{definition}
For brevity, the time variable $t$ will be omitted for the following analysis. For any given homeomorphic PPC error transformation function $\mathcal{R}\left(\cdot\right)$, by applying the shear mapping, $\mathcal{R}\left(\cdot\right)$ can be transformed into a new one, named as SMETF.
To distinguish with the traditional PPC error transformation function $\mathcal{R}\left(\cdot\right)$, we denote the SMETF as $\mathcal{R}_{s}\left(\cdot\right)$. Define the translated variable of $z_{si}$ calculated by the SMETF as $\boldsymbol{\varepsilon}_{s}\left(t\right) = \left[\varepsilon_{s1}(t),...\varepsilon_{si}(t)\right]^{\text{T}}\left(i=1,2,3\right)$, by applying the definition of shear mapping, the error transformation procedure of SMETF can be expressed as follows:
\begin{equation}
	\label{Rs}
	\varepsilon_{si} = \mathcal{R}_{s}\left(z_{si}\right) = \mathcal{R}\left(z_{si} - \varepsilon_{si} \tan\theta\right)
\end{equation}

Define a variable as $\boldsymbol{z}_{0} = [z_{01},...z_{0i}]^{\text{T}}\left(i=1,2,3\right)$, where its component is expressed as $z_{0i} = z_{si} - \varepsilon_{si}\tan \theta$. Further, denote its corresponding translated variable calculated by $\mathcal{R}(\cdot)$ as $\boldsymbol{\varepsilon}_{0} = [\varepsilon_{01},...\varepsilon_{0i}]^{\text{T}}\left(i=1,2,3\right)$ such that $\varepsilon_{0i} = \mathcal{R}\left(z_{0i}\right)$, we have the following result:
\begin{equation}
	\begin{bmatrix}
		z_{si} \\ \varepsilon_{si}
	\end{bmatrix} = 
	\begin{bmatrix}
		1 & \tan\theta \\
		0 & 1
	\end{bmatrix}
	\begin{bmatrix}
		z_{0i} \\ \varepsilon_{0i}
	\end{bmatrix}
\end{equation}

It can be observed that the explicit expression of $\mathcal{R}_{s}\left(\cdot\right)$ is hard to obtain, which indicates that the value of the translated variable $\varepsilon_{si}$ cannot be obtained directly. Nevertheless. note that the shear mapping will only squeeze the original function graph horizontally. Owing to this characteristic,  $\label{EQUAL} \varepsilon_{si} = \varepsilon_{0i}$ will be always hold.
This indicates us that the value of $\varepsilon_{si}$ can be obtained by calculating the value of $\varepsilon_{0i}$ instead. 

In summary, the main procedure of the calculation of the proposed SMETF can be stated as follow: 

$\boldsymbol{1}$.For any $z_{si}\in\left(-\infty,+\infty\right)$ calculated by the real-time system error state $e_{i}$ and RPF $\rho_{i}$ as $z_{si}\left(t\right) = e_i\left(t\right)/\rho_{i}\left(t\right)$, solve the following equation:
\begin{equation}
	\label{solve}
	z_{0i} + \mathcal{R}\left(z_{0i}\right)\tan\theta - z_{si} = 0
\end{equation}
Note that the value of $z_{si}$ is a known value for each control period and $z_{0i}$ is the variable to be solved. 

$\boldsymbol{2}$.Since the explicit expression of $\varepsilon_{0i}$ with respect to $z_{0i}$ is available, the value of $\varepsilon_{0i}$ can be calculated directly as $\varepsilon_{0i}=\mathcal{R}\left(z_{0i}\right)$. As a result, the value of $\varepsilon_{si}$ can be attained in this way.
A sketch map of this procedure is illustrated in Figure [\ref{figusm}].

\begin{figure}[hbt!]
	\centering 
	\includegraphics[scale=0.4]{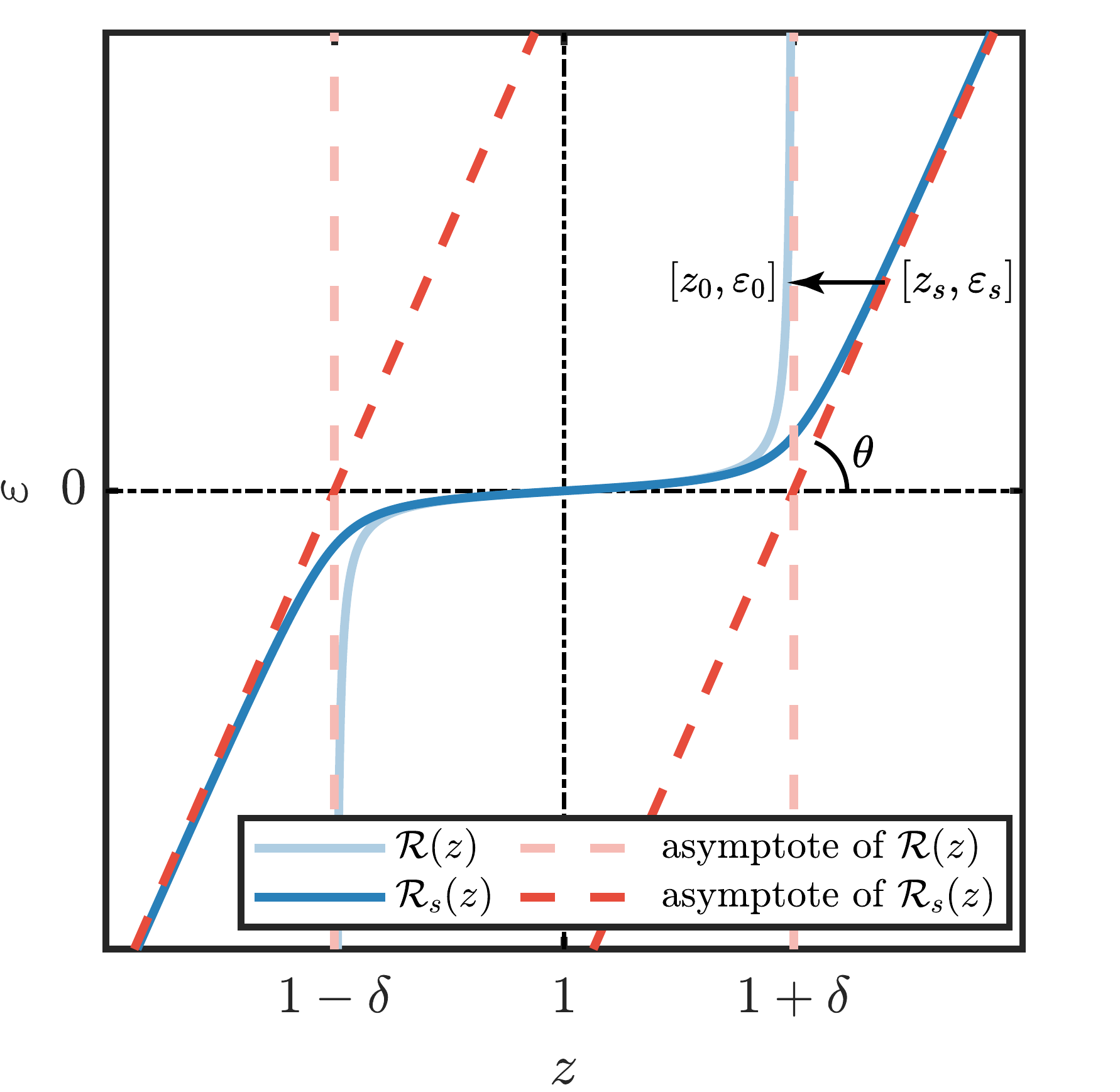}  \caption{Calculation Process of SMETF}  
	\label{figusm}
\end{figure}
In this paper, we choose the tangent function as the basic homeomorphic PPC error transformation function $\mathcal{R}\left(\cdot\right)$. The selected function $\mathcal{R}\left(\cdot\right)$ is stated as follows:
\begin{equation}
	\label{chosenfunction}
	\mathcal{R}\left(z_{0i}\right) = \tan\left[\frac{\pi}{2a}\left(z_{0i}-b\right)\right]
\end{equation}
where $a$, $b$ are the coefficients that needs designing. To cope with the state constraint elaborated in section IV.\ref{Constraint}, the coefficients are given as $a = \delta_{i}(t)$, $b = 1$. 
 Take the time-derivative of $\varepsilon_{si}$, we have:
\begin{equation}
	\label{derivativeofE}
	\dot{\varepsilon}_{si} = 
	\frac{\partial \varepsilon_{si}}{\partial z_{si}}\dot{z}_{si} + 
	\frac{\partial \varepsilon_{si}}{\partial \delta_{i}} \dot{\delta}_{i}
\end{equation}

Define $\mathcal{P}_{si} = \partial\varepsilon_{si}/\partial z_{si}\left(i=1,2,3\right)$ and $\mathcal{P}_{0i} = \partial\varepsilon_{0i}/\partial z_{0i}\left(i=1,2,3\right)$. Consider the first term in equation (\ref{derivativeofE}), by combining it with equation $\left(\ref{solve}\right)$, we have:
\begin{equation}
	\label{P0PS}
	\mathcal{P}_{si} = \frac{\mathcal{P}_{0i}}{1 + \mathcal{P}_{0i}\tan\theta}
\end{equation}
According to the chosen basic error transformation function $\mathcal{R}\left(\cdot\right)$ as (\ref{chosenfunction}), for any calculated $z_{0i}$, $\mathcal{P}_{0i}$ can be expressed as follows:
\begin{equation}
	\begin{aligned}
			\label{P0i}
		\mathcal{P}_{0i} &= \frac{\pi}{2\delta_{i}}\sec^{2}\left(\frac{\pi}{2\delta_{i}}\left(z_{0i}-1\right)\right) \\
		& = \frac{\pi}{2\delta_{i}}\left[\tan^{2}\left(\frac{\pi}{2\delta_{i}}\left(z_{0i}-1\right)\right)+1\right] \\
		&= \frac{\pi}{2\delta_{i}}\left(\varepsilon^{2}_{0i} +1\right) = \frac{\pi}{2\delta_{i}}\left(\varepsilon^{2}_{si} +1\right)
	\end{aligned}
\end{equation}
Substituting equation (\ref{P0i}) into (\ref{P0PS}), we have the following result:
\begin{equation}
	\label{PS}
	\mathcal{P}_{si} = \frac{\pi\left(\varepsilon^{2}_{si}+1\right)}{\pi\left(\varepsilon^{2}_{si}+1\right)\tan\theta + 2\delta_{i}}
\end{equation}

Consider about the second term in the equation \ref{derivativeofE} expressed as $\frac{\partial \varepsilon_{si}}{\partial \delta_{i}} \dot{\delta}_{i}$, since $\varepsilon_{si} = \varepsilon_{0i}$ will be always holds, we have the following property:
\begin{equation}
	\frac{\partial \varepsilon_{si}}{\partial \delta_{i}}\dot{\delta}_{i} 
	= \frac{\partial \varepsilon_{0i}}{\partial \delta_{i}}\dot{\delta}_{i}  = 
	\frac{\pi}{2}\left(\varepsilon^2_{si}+1\right) \left(z_{0i}-1\right)\left(-\frac{1}{\delta_{i}^{2}}\right)\dot{\delta}_{i} 
\end{equation}
where $\varepsilon_{0i}$ and $z_{0i}$ are the aforementioned corresponding original image of $\varepsilon_{si}$ and $z_{si}$ respectively such that $z_{0i} + \varepsilon_{0i}\tan\theta = z_{si}$. According to the aforementioned analysis in section IV.\ref{Constraint}, we have $\dot{\delta}_{i} \ge 0$. Note that for $z_{0i} \ge 1$, we have $\varepsilon_{si} \ge 0$, and for $z_{0i} \le 1$, we have $\varepsilon_{si} \le 0$ . Thus, we have the following conclusion:
\begin{equation}\label{SupControl}
	\begin{aligned}
		\begin{split}		
			\frac{\partial \varepsilon_{si}}{\partial \delta_{i}}\dot{\delta}_{i} \left \{	
			\begin{array}{ll}
				\ge 0,  & if z_{0i} \le 1\\				
				< 0, & if z_{0i} > 1
			\end{array}				
			\right.		
		\end{split}
	\end{aligned}
\end{equation}
By sort out these above results, the time-derivative of $\varepsilon_{si}$ can be rearranged as follows:
\begin{equation}\label{dEpsilon}
	\begin{aligned}
		\dot{\varepsilon}_{si} &= 
		\frac{\partial \varepsilon_{si}}{\partial z_{si}}\dot{z}_{si}+\xi_{si} \\
		&= \frac{\partial \varepsilon_{si}}{\partial z_{si}}\frac{\dot{e}_{i}\rho_{i} - e_{i}\dot{\rho}_{i}}{\rho^{2}_{i}} + \xi_{si} \\
		&= \psi_{si}\left(\dot{e}_{i} + \eta_{si}e_{i}\right) + \xi_{si}
	\end{aligned}
\end{equation}

where $\psi_{si} = \frac{1}{\rho_{i}}\mathcal{P}_{si}$, $\eta_{si}  = -\dot{\rho_{i}}/\rho_i$, $\xi_{si} = \frac{\partial \varepsilon_{si}}{\partial \delta_{i}}\dot{\delta}_{i}$.

\subsection{Analysis on SAPPC and Traditional PPC scheme}
\label{Comparison}

This section presents a through theoretical analysis on the effect between the proposed SAPPC scheme and the traditional PPC scheme. 
As stated in \cite{wei_2021_overview,hu_adaptive_2018}, the typical traditional PPC schemes can be classified into two types, and we name them as the "Homeomorhphic error transformation type PPC" and the "BLF type PPC" in this paper. 
The main characteristics of these two types of PPC error transformation procedure are listed in Table [\ref{TWOPPCw}] as below. 

% \usepackage{multirow}

% \usepackage{multirow}
% \usepackage{multirow}
\begin{table}[hbt!]
	
	\centering
	\begin{tabular}{|l|c|}
		\hline
		\multicolumn{1}{|c|}{Homeomorphic Type}                                       & BLF Type                                                                                                       \\ \hline
		$\varepsilon_{i} = \ln\left(\frac{K+z_{si}}{(K)\left(1-z_{si}\right)}\right)$ & \multirow{2}{*}{$\varepsilon_{i} = \frac{2e_{i} - \left(\rho_{ui} + \rho_{li}\right)}{\rho_{ui} - \rho_{li}}$} \\ \cline{1-1}
		$\varepsilon_{i} = \ln\left(\frac{K\left(1+z_{si}\right)}{K-z_{si}}\right)$   &                                                                                                                \\ \hline
		$-K\rho_{i}\left(t\right) < e_{i}\left(t\right) < \rho_{i}\left(t\right) $    & \multicolumn{1}{l|}{\multirow{2}{*}{$\rho_{li}\left(t\right) < e(t) < \rho_{ui}\left(t\right)$}}               \\ \cline{1-1}
		$-\rho_{i}\left(t\right) < e_{i}\left(t\right) < K\rho_{i}\left(t\right) $    & \multicolumn{1}{l|}{}                                                                                          \\ \hline
	\end{tabular}
\caption{Characteristic of Traditional PPC Error Transformation }
\label{TWOPPCw}
\end{table}

where $K\in\left(0,1\right)$ is a constant to be designed. The first row presents the error transformation of each scheme, while the second row presents the corresponding state constraint.

$\textbf{1. The singularity Problem}$

$\mathbf{1.}$ For the homeomorphic PPC error transformation,
note that these following properties will be satisfied for $e_{i}\left(0\right) \ge 0$:
\begin{equation}
	\begin{aligned}
				\lim_{z_{si}\to 1}\mathcal{R}\left(z_{si}\right) &= +\infty  \\
				\lim_{z_{si}\to -K}\mathcal{R}\left(z_{si}\right) &= -\infty		
	\end{aligned}
\end{equation}
This property indicates us that the definition domain of the homeomorphic type error transformation function is restricted in a belt-like region enclosed by two vertical asymptotes, thus the homeomorphic type error transformation function will be meaningless when $z_{si} \notin \left(-K,1\right)$. Therefore, the controller will trap into the singularity problem under such a condition.

$\mathbf{2}.$ For the BLF type PPC controller, it's state constraint is realized by the Barrier Lyapunov Function. Consider a typcial Barrier Lyapunov Function $V_{\text{BLF}}$ expressed as follows:
\begin{equation}
	V_{BLF} = \frac{k_{0}}{2}\left(\frac{1}{1-\varepsilon_{i}^{2}}\right)
\end{equation}
Take the time-derivative of $V_{BLF}$, we have:
\begin{equation}
	\dot{V}_{BLF} = \frac{k_{0}}{2}\left(\frac{1}{1-\varepsilon_{i}^{2}}\right)^{2}\left(2\varepsilon_{i}\dot{\varepsilon_{i}}\right)
\end{equation}
Consider about the condition that the state trajectory is out of the constraint region, i.e., $\varepsilon_{i} > 1$ or $\varepsilon_{i} < -1$, $V_{\text{BLF}} < 0$ will be hold, which indicates that the barrier lyapunov function is meaningless.
From another aspect, to guarantee the convergence of $\varepsilon_{i}$ under such a condition, $\varepsilon_{i}\cdot\dot{\varepsilon}_{i} < 0$ should be satisfied. Therefore, $\dot{V}_{\text{BLF}} < 0$ is hold for $V_{\text{BLF}} < 0$. This means that once the state trajectory is out of the constraint region, the BLF PPC-based controller will not able to lead the state trajectory back to the constraint region anymore. This phenomenon will be detailed and supported by the numerical simulation result later in the section V.\ref{compsimulation}.

\begin{remark}
	Although it is possible to technically guarantee that the state trajectory stays in the constraint region at $t=0$. However, it cannot be ensured that the state trajectory will stay in the constraint region during the whole control process, especially when there exists big external disturbances. This situation will be simulated at section V.\ref{compsimulation} for the detailed elaboration.
\end{remark}

\textbf{2. Chattering Caused by over-control problem} \\

For the convenience of the following analysis, we define the lower boundary and the upper boundary of the auxiliary variable $z_{si}$ as $U_{-}$ and $U_{+}$ respectively, i.e. the state constraint can be  expressed as $U_{-} < z_{si} < U_{+}$.

\begin{figure}[hbt!]
	\centering  
	\includegraphics[scale = 0.3]{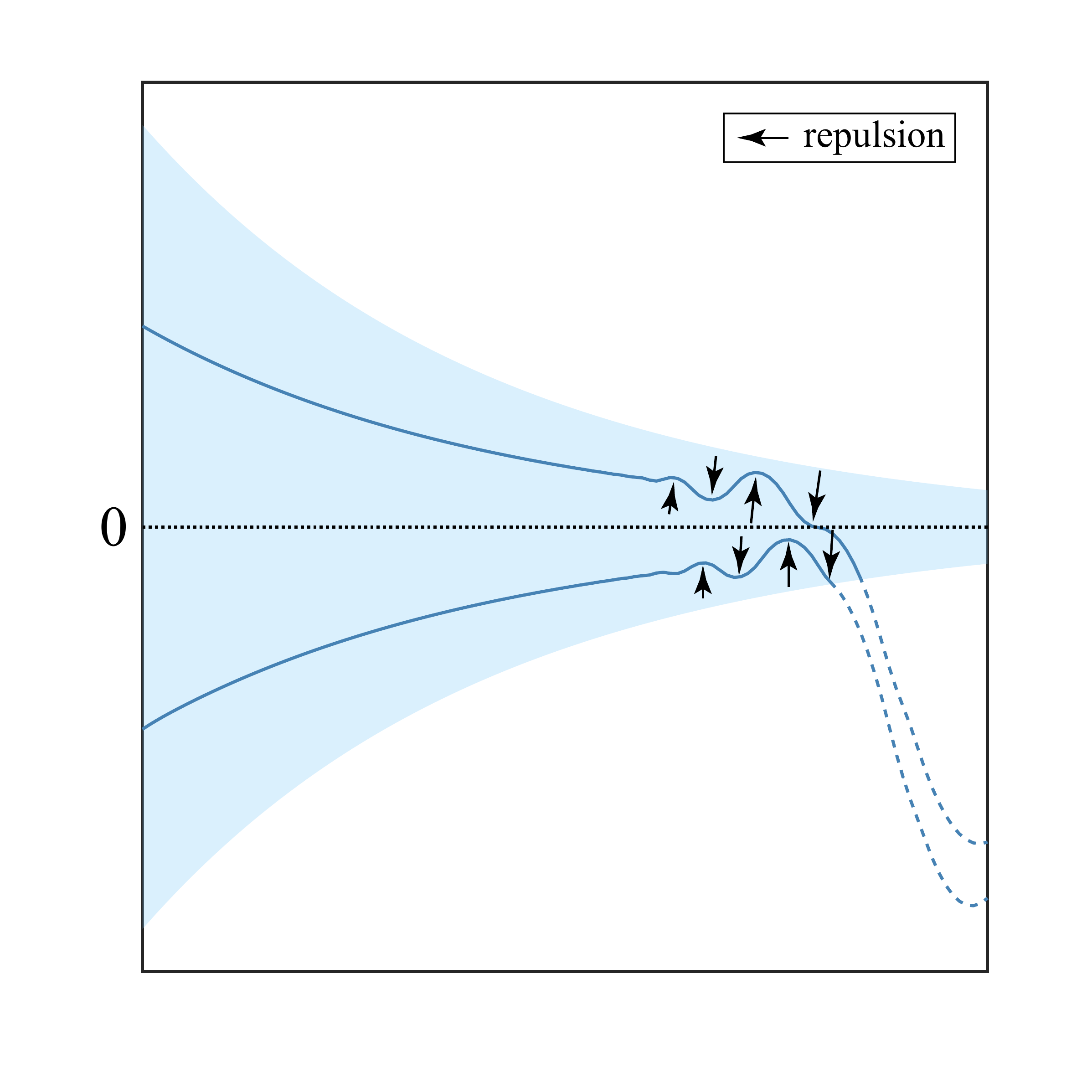}  \caption{The sketch map of the Chattering problem in Traditional PPC. Blue region denotes the constraint region, the arrow denotes the repulsive force exerted by the constraint boundary} 
	\label{fig_chat} 
\end{figure}

Accordingly, the controller will calculate a tremendous output when $z_{si} \to U_{+}$ or $U_{-}$ is hold. Due to this over control problem, the state trajectory will deviate from the desired path. 
Suppose the terminal value of the performance function is small at the steady-state. In that case, the constraint region will be narrow, and the state trajectory may alternately battered by the strong repulsive forces exerted by the constraint boundary. This will cause the state trajectory chattering intensely, as illustrated in Figure [\ref{fig_chat}].

\textbf{3. Contradiction between control accuracy and stability}

If the terminal value of the performance function is set to a small value, the constraint may be too tight for the auxiliary variable. This will cause the chattering problem as stated in the last subsection.

Oppositely, if the terminal value of the performance function is set to a bigger value, the constraint region will be wider. However, this will make the constraint strength too loose, which will result in a decrease in constraint ability. As a result, there is a contradiction between the system stabilization and control accuracy, limiting the effect of these two kinds of benchmark controllers. This will be validated in the numerical section. Due to this reason, we are not able to set the terminal value as small as we want.

\begin{remark}
	\label{remarkref}
	The steep constraint boundary for the state trajectory is the main sources of the PPC scheme's state constraint ability: when the state trajectory approaches the boundary of the constraint region as $z_{si}\to U_{+}$ or $z_{si}\to U_{-}$, the PPC-based controller will calculate a tremendous controller output to restrain the state responding trajectory, forcing it to stay in the constrained region.
	
	In this aspect, the steep constraint boundary exists in these error transformations cannot be removed directly.
\end{remark}

$\textbf{4. Analysis on SAPPC}$

As for the proposed SAPPC scheme, due to the utilization of the shear mapping, the graph of the original homeomorphic error transformation function will be tilted with a specific inclination degree decided by $\theta$. 
Accordingly, the original vertical asymptotes $z_{si} = U_{+}$ , $z_{si} = U_{-}$ will be tilted into two new asymptotes expressed as $tan\theta\left(z_{si}-U_{+}\right)$ and $tan\theta\left(z_{si}-U_{-}\right)$, shown as Figure [\ref{figusm}] In this way, the definition domain constraint is $\mathbf{loosed}$ by the shear mapping. Therefore, there exists no explicit boundary of the $z_{si}$, which ensures that the error transformation is able to perform gloablly.

Further, note that if $\theta = 0$, the SMETF $\mathcal{R}_{s}\left(\cdot\right)$ will be equivalent to the original homeomorphic error transformation function $\mathcal{R}\left(\cdot\right)$. By choosing an appropriate parameter $\theta$, the SAPPC-based controller will able to exert big enough constraint strength to the state trajectory.

\begin{remark}
	The effect of the shear mapping will actually loose the original constraint boundary in a rational way, and the angular parameter $\theta$ plays a significant role in this procedure. Although the smaller $\theta$ is, the stronger the constraint is, we find $\theta = 10^{\circ}$ is enough for the state trajectory constraint.
\end{remark}

%In conclusion, the acting mechanism of the proposed SAPPC can be stated as follows: Firstly, owing to the application of SMETF, the error transformation is able to perform globally non-singular, there exists a unique translated variable $\varepsilon_{si}$ for any given $z_{si}$. Therefore, the state trajectory will able to converge to the constraint region from anywhere in the state space, even when there exists strong disturbance. Secondly, the time-varying constraint boundary will provide a mild constraint at the steady-state, and this help alleviates the over-control problem since the constraint boundary is getting wider at that moment. This characteristic along with the first one ensures the guaranteed performance will be achieved finally. Thirdly, the maximum changing rate of the SMETF can be explicitly set through $\theta$, and this will alleviate the over-control problem.

	\subsection{Control Law Derivation}
\label{ControlLaw}
Like the traditional PPC scheme, the proposed SAPPC scheme can combine with various kinds of control methodologies to realize a specific controller. This subsection develops a backstepping controller based on the translated variable of attitude error quaternion $\boldsymbol{q}_{ev}$ by utilizing the proposed SAPPC scheme and the predefined-time stability theory. Besides, the dynamic surface control method is also employed to help alleviate the “differential explosion" problem. The system structure of the developed controller is illustrated in Figure [\ref{figSD}].

\begin{figure}[hbt!]
	\centering 
	\includegraphics[scale=0.017]{controlsystem_2}  \caption{Control System Diagram}  
	\label{figSD}  
\end{figure}

For further analyzing, the auxiliary variable of the attitude error quaternion's vector part is defined as $\boldsymbol{z}_{q} = \left[z_{q1},...z_{qi}\right]^{\text{T}}\left(i=1,2,3\right)$, where each component $z_{qi}$ is calculated by $z_{qi} = q_{evi}/\rho_{i}$. $\rho_{i}$ denotes the RPF trajectory that needs tracking, whose symbol is determined by the corresponding initial condition.  The translated variable of $\boldsymbol{z}_{q}$ is defined as $\boldsymbol{\varepsilon}_{q} = \left[\varepsilon_{q1},...\varepsilon_{qi}\right]\left(i=1,2,3\right)$, while each element of $\boldsymbol{\varepsilon_{q}}$ is given by the SMETF as $\varepsilon_{qi} = \mathcal{R}_{s}\left(z_{qi}\right)$.

According to the equation (\ref{derivativeofE}), we have the following result:
\begin{equation}
	\dot{\varepsilon}_{qi} =  \frac{1}{\rho_{i}}\frac{\partial \varepsilon_{qi}}{\partial z_{qi}} \left[\dot{q}_{evi} - \frac{\dot{\rho}_{i}}{\rho_{i}}q_{evi} \right] + \frac{\partial \varepsilon_{qi}}{\partial \delta_{i}}\dot{\delta}_{i}
\end{equation}

Define a diagonal matrix as $\boldsymbol{\psi}_{q}\in \mathbb{R}^{3\times 3}$, where $\psi_{qi} =  \frac{1}{\rho_{i}}\frac{\partial \varepsilon_{qi}}{\partial z_{qi}} \left(i=1,2,3\right)$.
 Similarly, define a diagonal matrix as $\boldsymbol{\eta}_{q}\in\mathbb{R}^{3 \times3}$ and a column vector as $\boldsymbol{\xi}_{q} = \text{vec}\left(\xi_{q1},\xi_{q2},\xi_{q3}\right)\in\mathbb{R}^{3}$, where $\eta_{qi} = \frac{-\dot{\rho}_{i}}{\rho_{i}}  \left(i = 1,2,3\right)$ and $\xi_{qi} = \frac{\partial \varepsilon_{qi}}{\partial \delta_{i}}\dot{\delta}_{i} \left(i = 1,2,3\right)$ respectively.
The equation can be rearranged into a matrix form, with the derivative of $\dot{\boldsymbol{\varepsilon}}_{q}$ can be expressed as:
\begin{equation}\label{dEpsilon}
	\dot{\boldsymbol{\varepsilon}}_{q}=\boldsymbol{\psi}_{q}\left(\dot{\boldsymbol{q}}_{e} + \boldsymbol{\eta}_q\boldsymbol{q}_e\right) + \boldsymbol{\xi}_q
\end{equation}

View the ideal error angular velocity as the virtual control law, denoted as $\boldsymbol{\alpha}$. For further analysis, we define three error subsystem as $\boldsymbol{z}_{1} = \boldsymbol{\varepsilon}_{q}$, $\boldsymbol{z}_{2} = \boldsymbol{\omega}_{e} - \boldsymbol{S}_d$ and $\label{Hd2} \boldsymbol{H}_{d} = \boldsymbol{S}_{d} - \boldsymbol{\alpha}$, where $\boldsymbol{S}_{d}$ represents the output of the filter. The $i$ th component of $\boldsymbol{z}_{1}$, $\boldsymbol{z}_{2}$, $\boldsymbol{H}_{d}$ are denoted as $z_{1i}$, $z_{2i}$ and $H_{di}$ in the following analysis, respectively.

$\mathbf{Step 1.}$ To guarantee that the translated variable $\boldsymbol{\varepsilon}_{q}$ is able to converge, choose a candidate Lyapunov Function $V_{1}$ as follows.
\begin{equation}
	V_1 = \frac{1}{2}\boldsymbol{z}_{1}^{\text{T}}\boldsymbol{z}_{1}
\end{equation}

Take the time-derivative of $V_{1}$ yields:
\begin{equation}\label{dVdV}
	\dot{V}_{1} = \boldsymbol{z}_{1}^{\text{T}}\dot{\boldsymbol{z}}_{1} = \boldsymbol{\varepsilon}^\text{T}_q \dot{\boldsymbol{\varepsilon}}_q
\end{equation}
Substituting (\ref{dEpsilon}) into (\ref{dVdV}), we can have:
\begin{equation}
	\label{dV1}
	\dot{V}_{1} = \boldsymbol{\varepsilon}^\text{T}_q\left[\boldsymbol{\psi}_{q}\left(\dot{\boldsymbol{q}}_{ev} + \boldsymbol{\eta}_q\boldsymbol{q}_{ev}\right) + \boldsymbol{\xi}_q\right]
\end{equation}

To ensure that $V_{1}$ will converge to the steady-state in a predefined-time, the ideal virtual control law $\boldsymbol{\alpha}$ is designed as follows:
\begin{equation}\label{virtualcontrol}
	\begin{aligned}
		\boldsymbol{\alpha} = \boldsymbol{\varGamma}^{-1}
		\left[-\boldsymbol{\psi}_q^{-1}\boldsymbol{M}_{q}\boldsymbol{K}_{q}\boldsymbol{\varepsilon}_q - \boldsymbol{\eta}_q\boldsymbol{q}_{ev}
		\right]
	\end{aligned}
\end{equation}
where $\boldsymbol{\varGamma}\left(\boldsymbol{q}_{e}\right)$ is expressed as $\boldsymbol{\varGamma}$ for brevity, $\boldsymbol{M}_{q}$, $\boldsymbol{K}_{q}$ are the diagonal matrix defined as follows:
\begin{equation}
	\begin{aligned}
		\boldsymbol{M}_{q} &= \frac{1}{2p_{1}T_{1}}e^{V_{1}^{p_{1}}}V_{1}^{-p_{1}}\cdot\boldsymbol{I_{3 \times 3}} \\
		\boldsymbol{K}_{q} &= K_{q}\cdot \boldsymbol{I_{3\times3}}
	\end{aligned}
\end{equation}

where $K_{q}$ represents the controller gain of the first layer, $p_{1}\in(0,1)$ and $T_{1} > 0$ are the coefficients to be designed.
Substituting (\ref{virtualcontrol}) into (\ref{dV1}), the expression can be reformed as follows:
\begin{equation}
	\dot{V}_{1} = -\boldsymbol{\varepsilon}^{\text{T}}_{q}\boldsymbol{M}_{q}\boldsymbol{K}_{q}\boldsymbol{\varepsilon}_{q} + \boldsymbol{\varepsilon}^{\text{T}}_{q}\boldsymbol{\xi}_{q}
\end{equation}
Applying the conclusion in (\ref{SupControl}), we can yield:
\begin{equation}
	\varepsilon_{qi}
	\begin{aligned}
		\begin{split}		
			\frac{\partial \varepsilon_{si}}{\partial \delta_{i}}\dot{\delta}_{i} \left \{	
			\begin{array}{ll}
				\le 0,  & since \quad z_{qi} \le 1, \varepsilon_{qi} \le 0          \\				
				\le 0, & since \quad z_{qi} \ge 1, \varepsilon_{qi} \ge 0  
			\end{array}				
			\right.		
		\end{split}
	\end{aligned}
\end{equation}	 
we can find that $ \boldsymbol{\varepsilon}^\text{T}_q \boldsymbol{\xi}_q \le 0$ will be always hold. Therefore, we have:
\begin{equation}
	\dot{V}_{1} \le -\boldsymbol{\varepsilon}^{\text{T}}_{q}\boldsymbol{M}_{q}\boldsymbol{K}_{q}\boldsymbol{\varepsilon}_{q} = -\frac{K_{q}}{p_{1}T_{1}}e^{V_{1}^{p_{1}}}V_{1}^{1-p_{1}}
\end{equation}
\begin{remark}
	To ensure that the Jacobian matrix $\boldsymbol{\varGamma}\left(\boldsymbol{q}_{e}\right)$ is invertible, the determination of $\boldsymbol{\varGamma}$ should not be zero. Note that $\text{Det}\left(\boldsymbol{2\varGamma}\right) = 2\cdot q_{e0} $, thus, $q_{e0}\left(t\right) \neq 0$ should be always satisfied. Since $ q_{e0} = 0$ means the attitude system is totally diverged, it is rational to indicate that the $ q_{e0}\left(t\right) \neq 0$ will be satisfied if we guarantee that the initial condition $ q_{e0}\left(0\right) \neq 0$ is satisfied.
\end{remark}

$\mathbf{Step 2.}$ In this step, the final control law is derived to ensure that the control error $\boldsymbol{z}_{2}$ will converge to a small enough residual set.
Consider the aforementioned $\boldsymbol{\omega}_{e}$ layer subsystem (\ref{errorsystem}), take the time-derivative of $\boldsymbol{z_{2}}$, we have:
\begin{equation}
	\begin{aligned}
		\label{dV2}
		\boldsymbol{J}\dot{\boldsymbol{z}_{2}} &= 	\boldsymbol{J}\left(\dot{\boldsymbol{\omega}}_{e} - 	\dot{\boldsymbol{S}}_{d}\right) \\
		& =  \boldsymbol{J}\boldsymbol{\omega}^{\times}_e\boldsymbol{C}_e\boldsymbol{\omega}_d 
		- \boldsymbol{J}\boldsymbol{C}_e\dot{\boldsymbol{\omega}}_d
		-\boldsymbol{\omega}_s^{\times}\boldsymbol{J}\boldsymbol{\omega}_s - \boldsymbol{J}\dot{\boldsymbol{S}}_{d} \\
		&\quad + \boldsymbol{u} + \boldsymbol{d}
	\end{aligned}
\end{equation}
Choose a candidate Lyapunov function as $V_{2} = \frac{1}{2}\boldsymbol{z}_{2}^{\text{T}}\boldsymbol{J}\boldsymbol{z}_{2}$.
To realize the fast tracking of the virtual control law $\boldsymbol{\alpha}$, the final control law is derived as follows:
\begin{equation}
	\begin{aligned}
			\label{U}
		\boldsymbol{u} &= -\boldsymbol{W}_{0} + \boldsymbol{J}\dot{\boldsymbol{S}_{d}} - 
		\boldsymbol{D}_{m}\text{vec}\left(\tanh\left(\frac{z_{2i}}{\mu_{i}}\right)\right)\\
		&\quad -\boldsymbol{M}_{\omega}\boldsymbol{K}_{\omega}\boldsymbol{J}\boldsymbol{z}_{2}
	\end{aligned}
\end{equation}
where $\boldsymbol{W}_{0} = \boldsymbol{J}\boldsymbol{\omega}^{\times}_e\boldsymbol{C}_e\boldsymbol{\omega}_d 
- \boldsymbol{J}\boldsymbol{C}_e\dot{\boldsymbol{\omega}}_d
-\boldsymbol{\omega}_s^{\times}\boldsymbol{J}\boldsymbol{\omega}_s$ denotes the dynamical terms in the equation. $\boldsymbol{D}_{m}$ represents a diagonal matrix consisted by the known upper boundary of the external disturbances, such that $\boldsymbol{D}_{m} = \text{diag}\left(D_{m} , D_{m}, D_{m}\right)$. $\mu_{i}\left(i=1,2,3\right)$ are constants to be designed, $\boldsymbol{M}_{\omega}$ and $\boldsymbol{K}_{\omega}$ are the diagonal matrix defined as follows:
\begin{equation}
	\begin{aligned}
		\boldsymbol{M}_{\omega} &= \frac{1}{2p_{2}T_{2}}e^{V_{2}^{p_{2}}}V_{2}^{-p_{2}} \cdot\boldsymbol{I}_{3 \times 3}\\
		\boldsymbol{K}_{\omega} &= K_{\omega}\cdot \boldsymbol{I_{3\times 3}}
	\end{aligned}
\end{equation}
where $p_{2} \in \left(0,1\right)$, $T_{2} > 0$ are coefficients to be designed later, $K_{\omega}$ represents the controller gain of the second layer system.
Take the time-derivative of $V_{2}$ and substituting (\ref{dV2}) and (\ref{U}) into the expression yields:
\begin{equation}
	\begin{aligned}
		\dot{V}_{2} &= \boldsymbol{z}_{2}^{\text{T}}\boldsymbol{J}\dot{\boldsymbol{z}}_{2}\\
		&= -\boldsymbol{z}_{2}^{\text{T}}\boldsymbol{M}_{\omega}\boldsymbol{J}\boldsymbol{K}_{\omega}\boldsymbol{z}_{2} + \boldsymbol{z}_{2}^{\text{T}}\boldsymbol{d} \\
		&\quad - \sum_{i=1}^{3}D_{\text{m}}z_{2i}\tanh\left(\frac{z_{2i}}{\mu_{i}}\right)\\
		& \le  -\frac{K_{\omega}J_{\text{min}}}{p_{2}T_{2}}e^{V_{2}^{p_{2}}}V_{2}^{1-p_{2}}
		+ \boldsymbol{z}_{2}^{\text{T}}\boldsymbol{d} \\
		&\quad - \sum_{i=1}^{3}D_{\text{m}}z_{2i}\tanh\left(\frac{z_{2i}}{\mu_{i}}\right)
	\end{aligned}
\end{equation}
Define $K_{2} = K_{\omega}J_{\text{min}}$, by applying the Lemma [\ref{lemma_tanh}] and assumption [\ref{assump_J}], we have the following result:
\begin{equation}
	\begin{aligned}
		\dot{V}_{2} &\le	 -\frac{K_{2}}{p_{2}T_{2}}e^{V_{2}^{p_{2}}}V_{2}^{1-p_{2}}  \\
		&\quad +
		D_{\text{m}}\sum_{i=1}^{3}\left[|z_{2i}| - z_{2i}\tanh\left(\frac{z_{2i}}{\mu_{i}}\right)\right] \\
		&\le -\frac{K_{2}}{p_{2}T_{2}}e^{V_{2}^{p_{2}}}V_{2}^{1-p_{2}} +
		0.2785D_{\text{m}}\sum_{i=1}^{3}\mu_{i}
	\end{aligned}    
\end{equation}

$\mathbf{Step 3.}$ In this step, we employ the dynamic surface control technique to built a filter for the $\dot{\boldsymbol{\alpha}}$, observing value of the virtual control law $\boldsymbol{\alpha}$. Choose a candidate Lyapunov function as $V_{3} = \frac{1}{2}\boldsymbol{H}_{d}^{\text{T}}\boldsymbol{H}_{d}$. Take the time-derivative of $V_{3}$, we have:
\begin{equation}
	\label{dV3}
	\dot{V}_{3} = \boldsymbol{H}_{d}^{\text{T}}\left(\dot{\boldsymbol{S}_{d}} - \dot{\boldsymbol{\alpha}}\right)
\end{equation}
the filter is designed as follows:
\begin{equation}
	\label{dSd}
	\dot{\boldsymbol{S}}_{d} = -\frac{1}{2p_{3}T_{3}}e^{V_{3}^{p_{3}}}V_{3}^{-p_{3}}\cdot \boldsymbol{I_{3\times3}}\cdot\boldsymbol{H}_{d}
\end{equation}
where $p_{3} \in \left(0,1\right)$ and $T_{3}$ are coefficients to be designed.
Substituting (\ref{dSd}) into (\ref{dV3}), one can be obtained that:
\begin{equation}
	\label{dV_3}
	\begin{aligned}
		\dot{V}_{3} &= \boldsymbol{H}_{d}^{\text{T}}\left[-\frac{1}{2p_{3}T_{3}}e^{V_{3}^{p_{3}}}V_{3}^{-p_{3}}\cdot\boldsymbol{I_{3\times3}}\boldsymbol{H}_{d}  - \dot{\boldsymbol{\alpha}}\right]\\
		&= -\frac{1}{2p_{3}T_{3}}e^{V_{3}^{p_{3}}}V_{3}^{1-p_{3}} + \boldsymbol{H}_{d}^{T}\left( - \dot{\boldsymbol{\alpha}}\right)\\
		&\le -\frac{1}{2p_{3}T_{3}}e^{V_{3}^{p_{3}}}V_{3}^{1-p_{3}} + 
		\|\boldsymbol{H}_{d}\| \|\dot{\boldsymbol{\alpha}}\|
	\end{aligned}
\end{equation}
By applying the Young's inequalities, note that:
\begin{equation}
	\begin{aligned}
		\|\boldsymbol{H}_{d}\|\| \dot{\boldsymbol{\alpha}}\|
		&\le \frac{1}{2}\|\boldsymbol{H}_{d}\|^{2} + \frac{1}{2}\|\dot{\boldsymbol{\alpha}}\|^{2} = V_{3}+ \frac{1}{2}\|\dot{\boldsymbol{\alpha}}\|^{2} \\
	\end{aligned}
\end{equation}
Note that $\dot{\boldsymbol{\alpha}}$ is a smooth continuous vector function all along. Consider the expression of $\boldsymbol{\alpha}$, according to (\ref{virtualcontrol}), 
take the time-derivative of $\boldsymbol{\alpha}$, we have:
\begin{equation}
	\begin{aligned}
		\dot{\boldsymbol{\alpha}} &= 
		\dot{\boldsymbol{\Gamma}^{-1}} \left[-\boldsymbol{\psi}_q^{-1}\boldsymbol{M}_{q}\boldsymbol{K}_{q}\boldsymbol{\varepsilon}_q - \boldsymbol{\eta}_q\boldsymbol{q}_{ev}
		\right] \\
		&\quad  
		-\boldsymbol{\Gamma}^{-1}\dot{\boldsymbol{\psi}^{-1}_{q}}\boldsymbol{M}_{q}\boldsymbol{K}_{q}\boldsymbol{\varepsilon}_{q} - \boldsymbol{\Gamma}^{-1}\boldsymbol{\psi}^{-1}_{q}\dot{\boldsymbol{M}_{q}}\boldsymbol{K}_{q}\boldsymbol{\varepsilon}_{q} \\
		&\quad -
		\boldsymbol{\Gamma}^{-1}\boldsymbol{\psi}^{-1}_{q}\boldsymbol{M}_{q}\boldsymbol{K}_{q}\dot{\boldsymbol{\varepsilon}}_{q}
		-\boldsymbol{\Gamma}^{-1}\dot{\boldsymbol{\eta}_{q}}\boldsymbol{q}_{ev} \\
		&\quad - \boldsymbol{\Gamma}^{-1}\boldsymbol{\eta}_{q}\dot{\boldsymbol{q}}_{ev}
	\end{aligned}
\end{equation}

For each part in $\dot{\boldsymbol{\alpha}}$, it is all bounded variable, which indicates that there exists a maxima of $\|\dot{\boldsymbol{\alpha}}\|$ during the whole control procedure. Further, we define $\boldsymbol{B}_{d} = \dot{\boldsymbol{\alpha}}$ for convenient.
Choose a candidate Lyapunov Function as $V = V_{1} + V_{2} + V_{3}$, take the time-derivative of $V$ yields:
\begin{equation}
	\begin{aligned}
		\dot{V} &= \dot{V_{1}}+\dot{V_{2}}+\dot{V_{3}}\\
		&\le -\frac{K_{q}}{p_{1}T_{1}}e^{V_{1}^{p_{1}}}V_{1}^{1-p_{1}}-
		\frac{K_{2}}{p_{2}T_{2}}e^{V_{2}^{p_{2}}}V_{2}^{1-p_{2}}\\
		&\quad-
		\frac{1}{p_{3}T_{3}}e^{V_{3}^{p_{3}}}V_{3}^{1-p_{3}}\\
		& \quad + 0.2785D_{\text{m}}\sum_{i=1}^{3}\mu_{i} + V_{3} + \frac{1}{2}\|\boldsymbol{B}_{d}\|^{2}
	\end{aligned}
\end{equation}
Since $\mu_{i}$ are the design parameter to be selected, so it is a constant during the whole procedure. Let $C_{0} = 0.2785D_{\text{m}}\sum_{i=1}^{3}\mu_{i}$ for convenient. Applying Lemma [\ref{lemma_last}], since $e^{x^{p}}x^{1-p} - x \le 0$ is always holds on $x\in\left(0,+
\infty\right)$, we have $V_{3} \le e^{V_{3}^{p_{3}}}V_{3}^{1-p_{3}}$, hence the expression can be rearranged into:
\begin{equation}
	\begin{aligned}
		\dot{V} 
		&\le -\frac{K_{q}}{p_{1}T_{1}}e^{V_{1}^{p_{1}}}V_{1}^{1-p_{1}}-
		\frac{K_{2}}{p_{2}T_{2}}e^{V_{2}^{p_{2}}}V_{2}^{1-p_{2}}\\
		&\quad -
		\left(\frac{1}{p_{3}T_{3}}-1\right)e^{V_{3}^{p_{3}}}V_{3}^{1-p_{3}}\\
		&\quad + C_{0}  + \frac{1}{2}\|\boldsymbol{B}_{d}\|^{2}
	\end{aligned}
\end{equation}
\begin{assumption}
	Consider a set $\Omega$ as $\Omega = \left\{V\le \mathcal{U} , \forall\mathcal{U} >0 \right\}$, with its boundary $\mathcal{U}$ can be set arbitrarily big. Therefore, for a specific $\mathcal{U}$, $\Omega$ is a compact set. According to assumption [\ref{assump_Qd}], the attitude desired trajectory set is embodied in a compact set as $\Omega_{d} = \left\{\|\boldsymbol{q}_{d}\|^{2} + \|\dot{\boldsymbol{q}_{d}}\|^{2} + \|\ddot{\boldsymbol{q}_{d}}\|^{2} \le \mathcal{U}_{d},\forall \mathcal{U}_{d}> 0\right\}$. Since the upper boundary $\mathcal{U}$ can be made large arbitrarily, it is assumed that $\Omega_{d} \subseteq \Omega$. 
\end{assumption}

Denotes the maxima of $\frac{1}{2}\|\boldsymbol{B}_{d}\|^{2}$ on $\Omega$ as $B_{m}$. Thus we yields:
\begin{equation}
	\begin{aligned}
			\label{dV}
		\dot{V} &\le -M_{a}
		\sum_{i=1}^{3}e^{V_{i}^{p_{i}}}V_{i}^{1-p_{i}}+C_{0}+B_{m}
	\end{aligned}
\end{equation}
where $M_{a}$ is expressed as:
\begin{equation}
	M_{a} = \min\left(\frac{K_{q}}{p_{1}T_{1}},
	\frac{K_{2}}{p_{2}T_{2}},
	\left(\frac{1}{p_{3}T_{3}}-1\right)\right)
\end{equation}
For convenient, we set $p_{1} = p_{2} = p_{3} = p$ in this paper for the following analysis, thus, the expression of $M_{a}$ can be expressed as follows:
\begin{equation}
	M_{a} = \min\left(\frac{K_{q}}{pT_{1}} , \frac{K_{2}}{pT_{2}} , \frac{1}{pT_{3}} -1\right)
\end{equation}
\begin{remark}
	For the synthesize of the controller, $M_{a} > 0$ should be always satisfied, thus $p_{3}T_{3}\in\left(0,1\right)$ should be guaranteed. Consider the dynamic filter system, since it is not a real physical system, so its parameter $T_{3}$ can be set to a far more small value compared with $T_{1}$ and $T_{2}$. This indicates that $M_{a}$ is mostly chosen according to $\frac{K_{q}}{T_{1}}$ or $\frac{K_{2}}{T_{2}}$, not the dynamic surface filter.
\end{remark}

Consider the first term in equation (\ref{dV}), assume that $V_{1}\le V_{2}\le V_{3}$. consider the function $f\left(x\right) = e^{x^p}\left(0<p<1\right)$, note we have $f(V_1)\le f(V_2) \le f(V_3)$. Similarly, consider the function $h\left(x\right) = x^{1-p}$, we have $h(V_1) \le h(V_2) \le h(V_3)$. The proof of this property is detailed in VII.\ref{appendixA}. By applying the Chebyshev's inequalities as Lemma [\ref{lemma_chebyshev}] , we have the following result:
\begin{equation}
	\sum_{j=1}^{3}e^{V_{j}^p}V_{j}^{1-p} \ge \frac{1}{3}\sum_{j=1}^{3}e^{V_{j}^p}\sum_{j=1}^{3}V_{j}^{1-p}
\end{equation}
subsequently, by applying Lemma [\ref{lemma_mean}], one can be obtained that:
\begin{equation}
	\frac{1}{3}\sum_{j=1}^{3}e^{V_{j}^p} \ge \left(\prod_{j=1}^{3}e^{V_{j}^p}\right)^{\frac{1}{3}} = e^{\frac{1}{3}\sum_{j=1}^{3}V_{j}}
\end{equation}
simultaneously, applying Lemma [\ref{lemma_beta}] yields
\begin{equation}
	\sum_{j=1}^{3}V^{1-p}_{j} \ge \left(\sum_{j=1}^{3}V_{j}\right)^{1-p} = V^{1-p}
\end{equation}
therefore, we have the following result:
\begin{equation}
	\begin{aligned}
			\label{results}
		-\sum_{j=1}^{3}e^{V_{j}^{p}}V_{j}^{1-p}  &\le -e^{\frac{1}{3}\left(\sum_{j=1}^{3}V_{j}\right)^{p}}\left(\sum_{j=1}^{3}V_{j}\right)^{1-p} \\
		& = -e^{\frac{1}{3}V^p}V^{1-p}
	\end{aligned}
\end{equation}
Substituting (\ref{results}) into (\ref{dV}), since $M_a \ge 0$, we have:
\begin{equation}
	\begin{aligned}
			\label{result}
		&-M_{a}\left[\sum_{j=1}^{3}e^{V_{j}^{p}}V_{j}^{1-p}\right] \\
		\quad &\le -M_{a}e^{\frac{1}{3}\left(\sum_{j=1}^{3}V_{j}\right)^{p}}\left(\sum_{j=1}^{3}V_{j}\right)^{1-p} \\
		& = -M_{a}e^{\frac{1}{3}V^p}V^{1-p}
	\end{aligned}
\end{equation}
Sort out these results, we can yield:
\begin{equation}
	\label{Sortout}
	\dot{V} \le -M_{a}e^{\frac{1}{3}V^{p}}V^{1-p}+C_{0}+B_{m}
\end{equation}
Note that we have select $p_{1} = p_{2} = p_{3} =p$, consider about the time parameter in $M_{a}$, rearrange it as follows:
\begin{equation}
	\begin{aligned}
		\frac{K_{q}}{p_{1}T_{1}} = \frac{1}{p\frac{T_{1}}{K_{q}}};
		\frac{K_{\omega}}{p_{2}T_{2}} = \frac{1}{p\frac{T_{2}}{K_{2}}};
		\frac{1}{pT_{3}}-1 = \frac{1}{p\frac{T_{3}}{1-pT_{3}}}
	\end{aligned}
\end{equation}
Define $T_{s1} = \frac{T_{1}}{K_{q}}$, $T_{s2} = \frac{T_{2}}{K_{2}}$, $T_{s3} = \frac{T_{3}}{1-pT_{3}}$. we can yield that $M_{a} = \frac{1}{p\max\left(T_{s1},T_{s2},T_{s3}\right)}$.
Further, we define $T_{s} = \max\left(T_{s1},T_{s2},T_{s3}\right)$ for brevity. 
\begin{remark}
	Here we discussing about the approximation of  $B_{m}$. According to the expression and the characteristic of $\boldsymbol{B}_{d}$, there exists an upper boundary of $0.5\|\boldsymbol{B}_{d}\|^{2}$. According to the attitude error system, we have the following result:
	\begin{equation}
		\dot{\boldsymbol{\alpha}} = \boldsymbol{J}^{-1}u_{d} + \boldsymbol{W}
	\end{equation} 
	where $\boldsymbol{W}$ is the known bounded dynamical terms, $\boldsymbol{u}_{d}$ denotes the desired control output according to the inverse dynamic equation.
	Therefore, the derivative of the desired attitude error angular velocity can be approximated by $\boldsymbol{u}_{d}$. In this way, the maxima of $\boldsymbol{B}_{d}$ is mainly reached at the beginning of the control process since the system will acting intensely to ensure the convergence of the system at that time. This indicates that we could approximate $B_{m}$ through the initial condition. Thus, we could have an approximate of $C_{0}+B_{m}$, which can be used for the parameter regulating.
\end{remark}

For the conclusion expressed in equation (\ref{Sortout}), by applying the conclusion in theorem \ref{theoremone}, we have
\begin{equation}
	\dot{V} \le -\mu M_{a}e^{\frac{1}{3}V^{p}}V^{1-p}
\end{equation} 
for $\exp\left(\frac{1}{3}V^{p}\right)V^{1-p} \ge \frac{C_{0} + B_{m}}{M_{a}\left(1-\mu\right)}$. For $\mu \in\left(0,1\right)$, the system will be practically predefined-time stabled, with the upper boundary of the settling time expressed as
\begin{equation}
	\label{settletime}
	T_{\text{set}} \le 3\frac{1}{\mu pM_{a}} = 3\frac{1}{\mu}T_{s}
\end{equation}
further, the residual set is expressed as follows:
\begin{equation}
	\left\{\boldsymbol{x}\big|\exp\left(\frac{1}{3}V^{p}\right)V^{1-p} \le
	\frac{(C_{0} + B_{m}) p T_{s}}{1-\mu}\right\}
\end{equation}

Note that for arbitrary $V > 0$, $\dot{V} < 0$ will be strictly satisfied. Thus, consider the compact set $\Omega$, it should be noted that once $V\left(0\right) \le \mathcal{U}$ is satisfied, $V \le \mathcal{U}$ will be satisfied. In this way, the compact set $\mathcal{U}$ is a invariant set of the system, and $\frac{1}{2}\|\boldsymbol{B}_{d}\|^{2} \le B_{m}$ will be always hold. All the closed-loop signals will be ultimately uniformly bounded. 
\begin{remark}
	Here we post some issues about the parameter selecting. Considering about a residual set of the lyapunov function $V$ at the terminal stage denoted as $V_{\infty}$, thus $V \le V_{\infty}$ is expected for $T \ge T_{\text{set}}$. For instance, let $p=0.1$, considering the expression of $V = \frac{1}{2}\left[\|\boldsymbol{\varepsilon_{q}}\|^{2} + \boldsymbol{z}_{2}^{\text{T}}\boldsymbol{J}\boldsymbol{z}_{2} + \|\boldsymbol{H}_{d}\|^{2}\right]$, we can set $V_{\infty}$ by evaluating the terminal desired status of the system. Hence we set $V_{\infty} = \frac{1}{2}\left[1+J_{max}\epsilon+\epsilon\right] \approx \frac{1}{2}$, where $\epsilon$ is a small enough positive constant. Accordingly, the value of $\exp\left(\frac{1}{3}V_{\infty}^{p}\right)V_{\infty}^{1-p}$ is 0.7313. This indicates us that the parameter selecting should guarantee that $\frac{\left(C_{0}+B_{m}\right)p\max\left(T_{s1},T_{s2},T_{s3})\right)}{1-\mu}\le 0.7313$.
	
	 For instance, take $\mu = 0.5$ for the consideration, assume that $C_{0} + B_{m} = 0.2785\cdot1\cdot3 + \frac{1}{J_{\min}} * 2 = 1.3355$, we have $p \max\left(T_{s1},T_{s2},T_{s3}\right) \le 2.7$, and the settling time will be $T_{\text{set}} \le 16.2s$. This result indicate us that the system will converge into the desired residual set $V_{\infty} \le 0.5$ before $T_{\text{set}} = 16.2s$. Note that the smaller the desired residual set is, the smaller the time parameter should be. Besides, according to the suggestion of the parameter selecting in \cite{sanchez_class_2018}, the final $p\max\left(T_{s1}, T_{s2}, T_{s3}\right) \le 0.5$ is recommended.
	
\end{remark}

\section{Numerical Simulation Results and Analysis}
\label{simulation}

In this section, several groups of simulation results are proposed to evaluate the effect of our proposed SAPPC scheme. A virtual attitude tracking task is established, with some performance requirements expected to satisfy. First, a group of attitude tracking control simulation results are shown to validate the applicability of the proposed scheme. Further, a comparison simulation is presented, with two typical PPC scheme-based controllers considered as a benchmark. Besides, a comparison simulation is illustrated to show the ability of the proposed scheme to handle sudden severe disturbances and the singularity problem. Finally, a Monte-Carlo simulation is carried out to show that the proposed scheme is universally applicable to arbitrary initial conditions.

In the simulation, the satellite is assumed to be a rigid body spacecraft, whose inertial matrix is defined as $\boldsymbol{J} = \text{diag}\left(4,4,4\right)\left(kg\cdot m^{2}\right)$. The maximum controller output on each axis is assumed to be 0.5$N\cdot m$, and the minimum controller output on each axis is assumed to be 0.005$N \cdot m$.

\subsection{Simulation of An Assumed Attitude Tracking Task}
\label{Simualtionpartone}
In this section, the spacecraft is expected to track a desired attitude trajectory denoted as $\boldsymbol{q}_d\left(t\right)$. The desired attitude trajectory is defined by the initial attitude quaternion $\boldsymbol{q}_{d}\left(0\right)$ and the desired attitude angular velocity $\boldsymbol{\omega}_{d}\left(t\right)$, expressed as follows:
\begin{equation}
	\begin{aligned}
			\boldsymbol{q}_{d}\left(0\right) &= [0,0,0,1]^{\text{T}}  \\ 
		\boldsymbol{\omega}_{d}\left(t\right) &= 0.573 [\cos\left(t/40\right),\sin\left(t/30\right),-\cos\left(t/50\right)]^{\text{T}}
	\end{aligned}
\end{equation}
To guarantee that the simulation case is convincing, we randomly choose the initial condition of the attitude quaternion. The initial condition of the spacecraft's attitude is randomly selected as follows:
\begin{equation}
	\begin{aligned}
			\boldsymbol{q}_{s}\left(0\right) &= [0.3254,0.4068,-0.3254,0.7891]^{\text{T}} \\ \quad
		\boldsymbol{\omega}_{s}\left(0\right) &= [0,0,0]^{\text{T}}
	\end{aligned}
\end{equation}
Further, the disturbance model is expressed as follows:
\begin{equation}
	\boldsymbol{d} = 
	\begin{bmatrix}
		0.001 \cdot \left[4\sin\left(3\omega_{p}t\right) + 3\cos\left(10\omega_{p}t\right) -20\right]\\ 	
		0.001\left[-1.5\sin\left(2\omega_{p}t\right) + 3\cos\left(5\omega_{p}t\right) +20\right]\\ 	
		0.001\left[3\sin\left(10\omega_{p}t\right) - 8\cos\left(4\omega_{p}t\right) +20\right]\\ 	
	\end{bmatrix}
\end{equation}
where $\omega_{p} = 0.01$ is a parameter used to determine its period.

As aforementioned above, the controller is expected to satisfy some performance requirements, stated as follows:

$\mathbf{1.}$ The attitude control error should converge to satisfy $\max\left(|q_{evi}|\right) < 1e-3$ in no more than $t=20s$. 
$\mathbf{2.}$ Terminal control error of the attitude quaternion should be no more than $0.02^{\circ}$ for the accuracy of the task, which is about to equivalent to $\max\left(|q_{evi}|\right) \le 1.1e-4$.

To satisfy these performance requirements, we first design a RPF for the error state trajectory. The coefficients of the RPF are selected as follows, shown in Table (\ref{tab:RPFpara}).

\begin{table}[hbt!]
	\centering
	\begin{tabular}{|l|l|}
		\hline
		Initial Value of Exponential Function Part $\rho_{e0}$  &  0.4\\ \hline
		Terminal Value of Exponential Function Part $\rho_{e\infty}$ &  1e-6 \\ \hline 
		Coefficient of Exponential Function Part $l$ & 0.5 \\ \hline
		Convergence Time of the RPF   $t_{2}$              & 20 \\ \hline
		Terminal Value of the RPF   $g_{\infty}$         & 3e-5 \\ \hline
	\end{tabular}
	\caption{\label{tab:RPFpara} Coefficients for the designed RPF}
\end{table}

 Further, we set the time parameter as $p = 0.1, T_{1} = T_{2} = 3s, T_{3}=2s$. This will ensure that deviation between the RPF and the state trajectory will converge before the settling time instant $t = 20s$.
\begin{figure}[hbt!]
	\centering 
	\includegraphics[scale = 0.68]{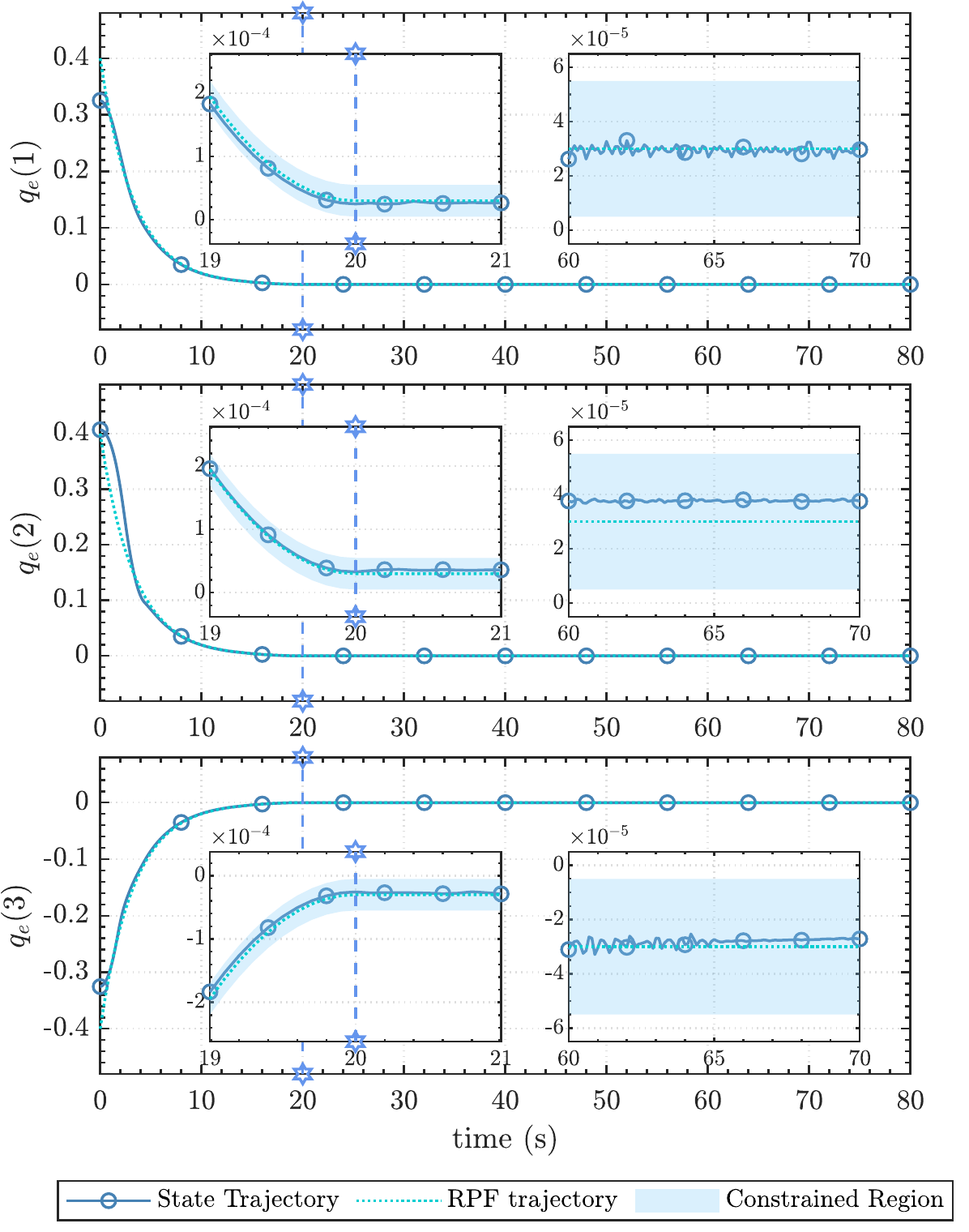}
	\caption{Time Responding of $\boldsymbol{q}_{evi}\left(i=1,2,3\right)$ (Normal Case)}    
	\label{fig_qe}  
\end{figure}

The simulation results are presented in Figure [\ref{fig_qe}][\ref{fig_Uoutput}]. Figure [\ref{fig_qe}] shows the evolution of the attitude error quaternion $\boldsymbol{q}_{ev}$ during the whole process, with each component illustrated respectively. Figure [\ref{fig_Uoutput}] shows the evolution of the actuator output of each axis. In Figure [\ref{fig_qe}], the solid line denotes the actual error state trajectory, while the dotted line denotes the RPF curve. The blue-filled region is the constraint region in which the state trajectory is expected to stay, and the blue star marker denotes the preassigned system settling time $t = 20s$.  We can find the state trajectory $q_{evi}$ rapidly converges to the RPF and finally reaches the steady-state under the guidance of the RPF, without any overshoot occurring.
Specifically, the convergence time of the error state is exactly $t = 20s$, and the terminal control error is bounded by  $\max\left(|q_{evi}|\right) \le 4e-5$.The deviation between the RPF and the state trajectory converged into a small set before $T_{s} = 15s$, which indicates that the practically predefined-time stable is achieved. These simulation results indicate that all the required performance requirements could be achieved.
\begin{figure}[hbt]
	\centering  
	\includegraphics[scale = 0.68]{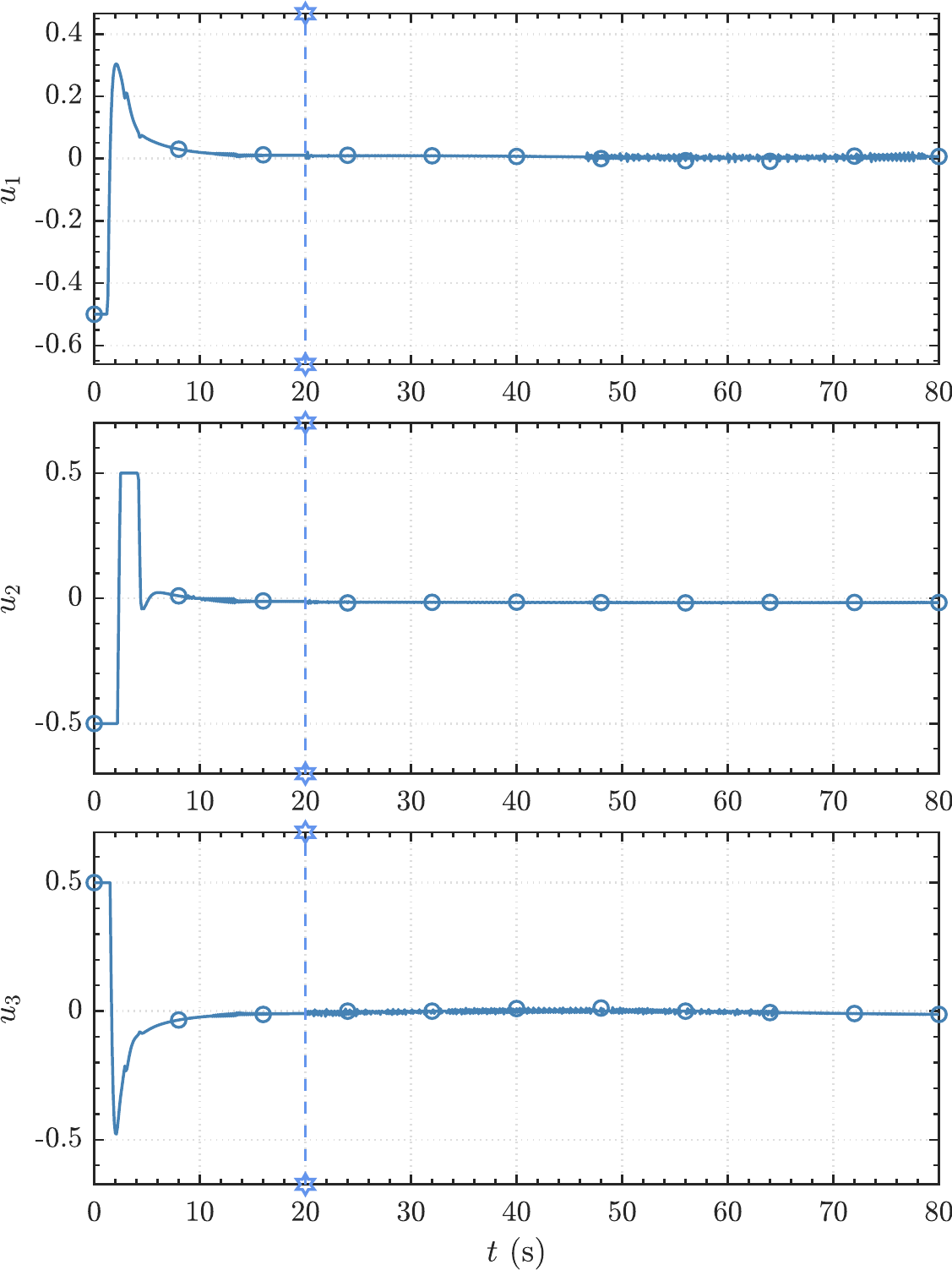}  \caption{Time Responding of actuator output $u_{i}\left(i=1,2,3\right)$ (Normal Case)} 
	\label{fig_Uoutput} 
\end{figure}

\subsection{Comparison Simulation of Benchmark Controller and SAPPC controller }
In this subsection, another two kinds of PPC structure are considered to refer to as a comparison. The PPC scheme of these two benchmark controllers is proposed in the literature \cite{wei_2021_overview} and \cite{hu_adaptive_2018}, respectively. We build two benchmark controllers correspondingly based on these two kinds of PPC schemes. The derived control law of these benchmark controllers are detailed in the Appendix VII.\ref{TraPPC} and VII.\ref{BLFPPC}.

In this simulation, the attitude tracking control task, spacecraft physical parameters and the external disturbance model are the same as the one used in subsection V.\ref{Simualtionpartone}. The initial condition of the spacecraft is randomly selected again, expressed as follows:
\begin{equation}
	\begin{aligned}
			\boldsymbol{q}_{s}\left(0\right) &= [0.2,-0.15,-0.25,0.9354]^{\text{T}} \\ \quad
		\boldsymbol{\omega}_{s}\left(0\right) &= [0,0,0]^{\text{T}}
	\end{aligned}
\end{equation}

We name the benchmark controller $1$ and the benchmark controller $2$ as the "TraPPC" and "BLFPPC" correspondingly. Figure [\ref{fig_qe_compare}] shows the $q_{evi}\left(t\right)$ trajectory of each controller respectively, while Figure [\ref{fig_u_compare}] shows the actuator output of each controller separately.

\begin{figure}[hbt!]
	\centering  
	\includegraphics[scale = 0.5]{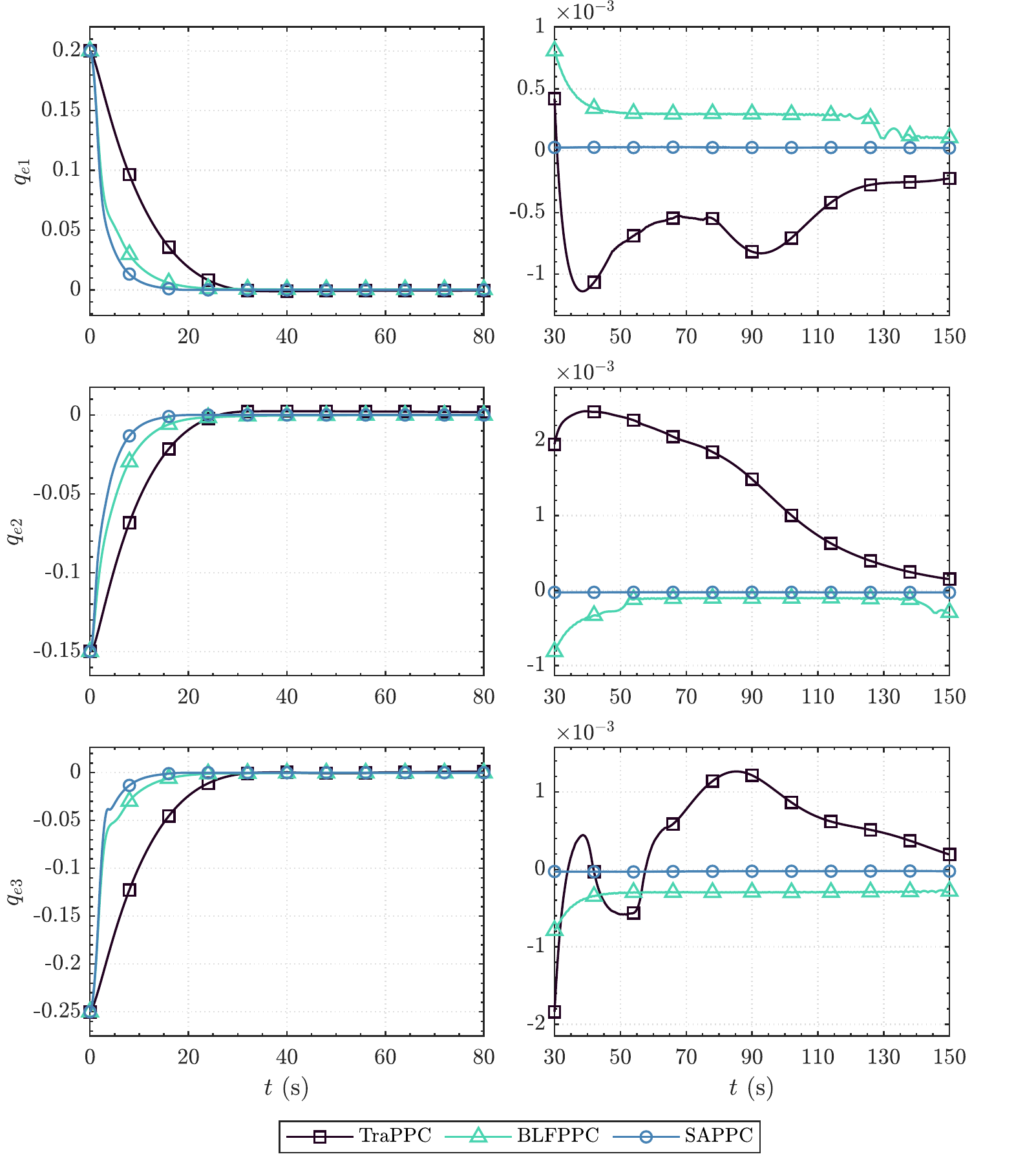}  \caption{Time Responding of attitude error quaternion $q_{evi}\left(i=1,2,3\right)$(Comparison)} 
	\label{fig_qe_compare} 
\end{figure}

\begin{figure}[hbt!]
	\centering  
	\includegraphics[scale = 0.6]{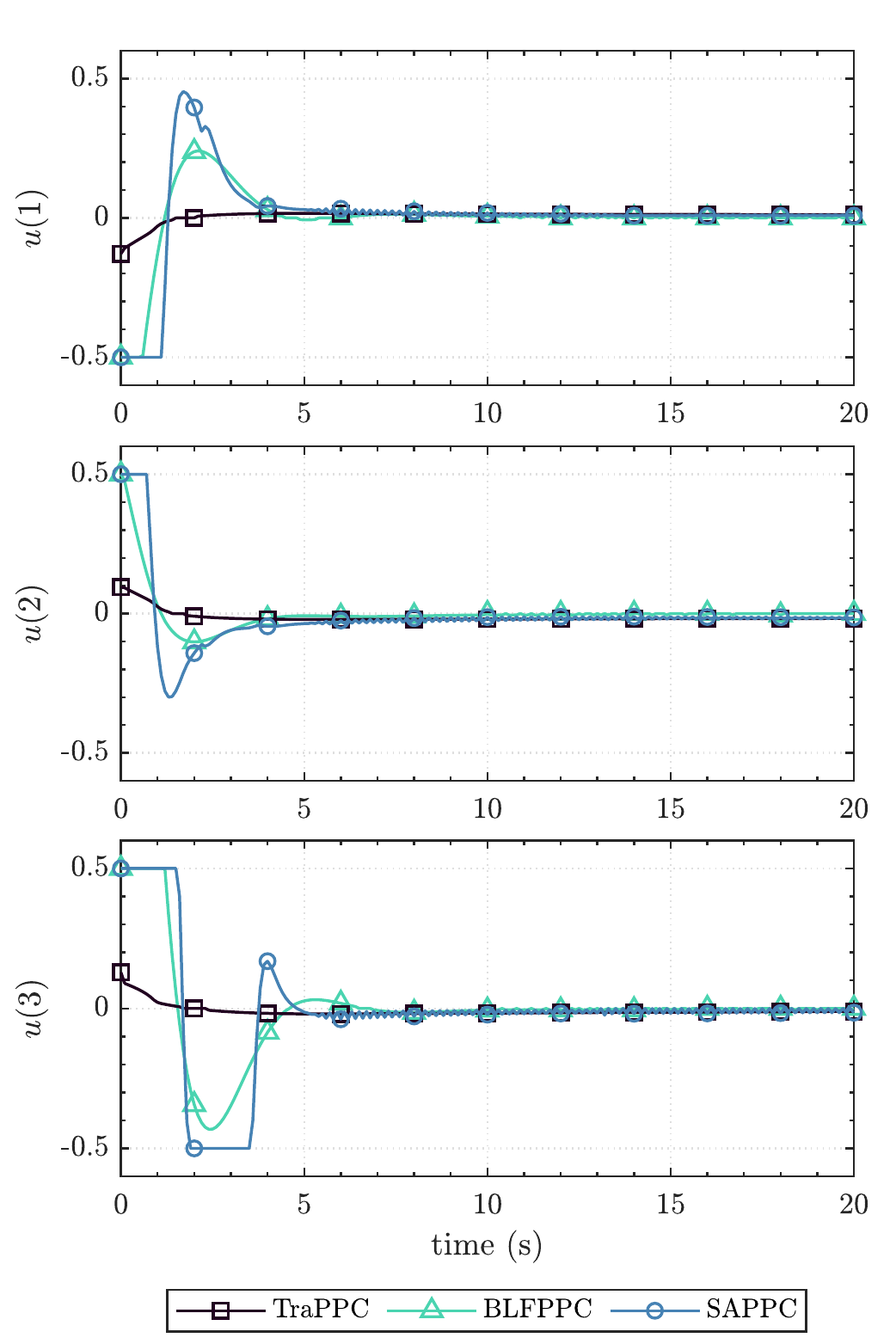}  
	\caption{Time Responding of actuator output $u_{i}\left(i=1,2,3\right)\left(Comparison\right)$} 
	\label{fig_u_compare} 
\end{figure}  
We can find that our proposed SAPPC scheme achieves the highest control accuracy among these controllers under the same conditions, as the control error is bounded to $\max\left(|q_{evi}\left(t\right)| \right)< 4.5e-5$. Also, we can observe that there's overshoot occurs in the TraPPC controller, while there's no overshoot occurs in the proposed scheme and the BLFPPC controller. 
\begin{remark}
	     By choosing a smaller terminal value for the performance function, the traditional PPC scheme may able to achieve a higher accuracy. However, this will bring the risk suffering from the chattering problem. As we test many times, the current result shown in the figure is probably the best result we can ever achieve.
\end{remark}

\subsection{Validation of the Singularity Avoidance} 
\label{compsimulation}
This section concentrates on the issue whether the SAPPC scheme is able to avoid the singularity or not. In this section, the state trajectory will be pulled out of the constraint region by a sudden disturbance. If the singularity-avoidance is realized, the state trajectory will be able to converge back to the constraint region. Besides, the BLFPPC controller is selected to be the reference controller in this subsection.  

The disturbance model is consisted by two kinds of disturbances. For $t<50s$, the disturbance model is the same as the one used in subsection V.\ref{Simualtionpartone}. At $t=50s$, an additional sudden severe disturbance modeled as $\boldsymbol{d}_{a} = [1,1,1]^{\text{T}}\left(N\cdot m\right)$ will be exerted to the system, last for $0.5s$. 

\begin{figure}[hbt!]
	\centering  
	\includegraphics[scale = 0.55]{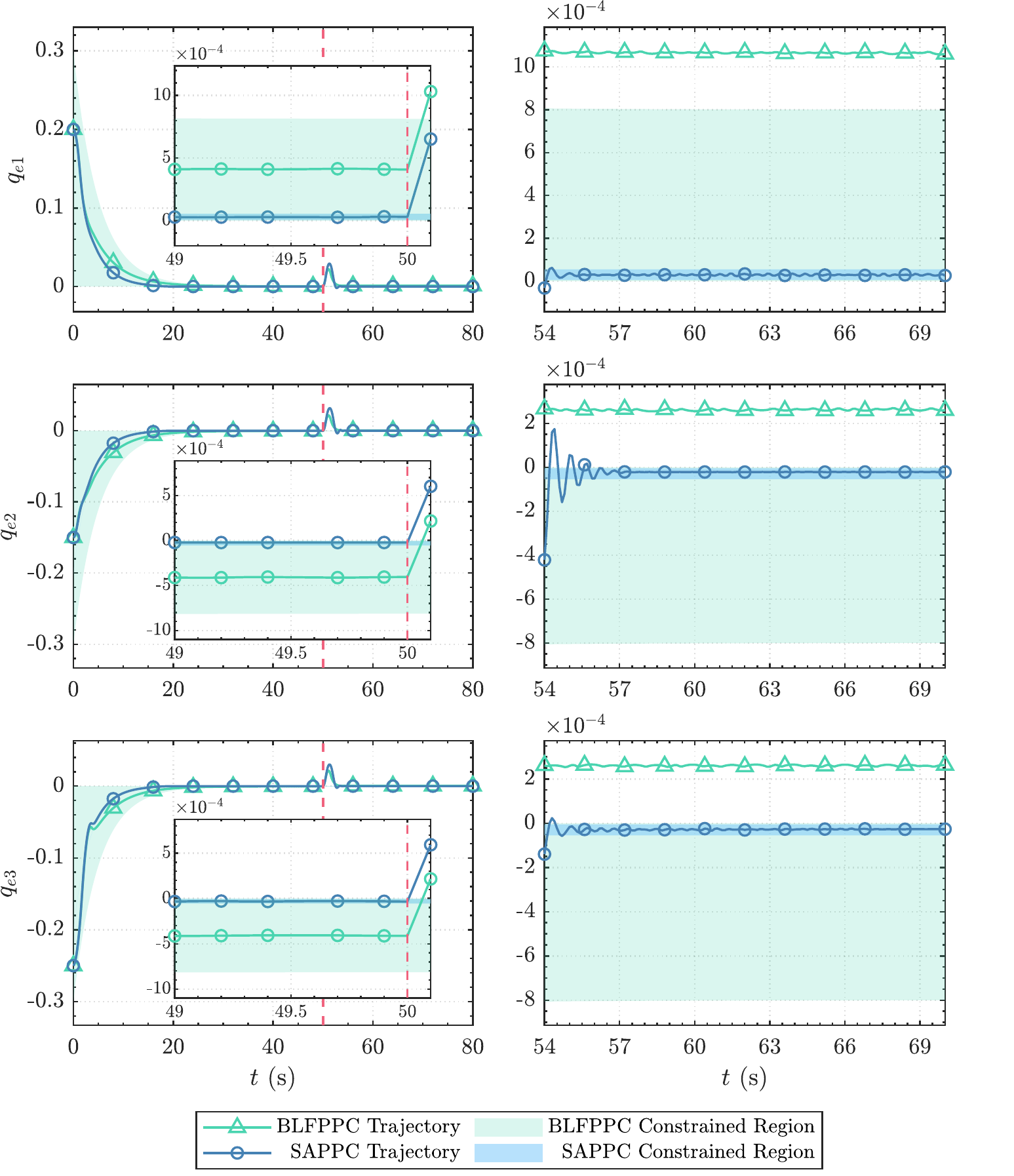}  \caption{Time responding of $q_
		{evi}$ under sudden disturbance} 
	\label{fig_sudden} 
\end{figure}

The simulation results of the $q_{evi}$ trajectory and the controller output are illustrated in Figure[\ref{fig_sudden}]. The state trajectory and the constraint region of SAPPC is illustrated in blue, while the BLFPPC controller is illustrated in green.

we can observe that when the sudden disturbance is exerted to the system at $t=50s$, both the state trajectory of BLFPPC and SAPPC are pulled out of the constraint region by the disturbance. However, the state trajectory of BLFPPC benchmark controller can no longer converge back to the green constraint region. As for the proposed SAPPC scheme, the state trajectory is still able to converge back to the constraint region even when the system suffered from sudden disturbance. This indicates that the singularity problem is solved in the SAPPC scheme.

\begin{remark}
	We can observe that the "BLFPPC" benchmark controller will not diverge totally after the disturbance, but the system is not able to converge back to the constraint region. Here we give an explanation for this phenomenon.
	
	When the state trajectory is pulled out of the constraint region by the external disturbance at $t=50s$, note that the third term expressed as $- 2\boldsymbol{K}_{3}\boldsymbol{\Gamma}^{-1}\boldsymbol{D}_{\rho}\boldsymbol{\varepsilon}_{q}$ will exert the wrong direction repulsive force due to the singularity problem, as stated in Section IV.\ref{Comparison}.
	Considering about the control law of the BLFPPC benchmark controller expressed in section VII.\ref{BLFPPC}, the proportion term $-K_{2}\boldsymbol{M}_{2}\boldsymbol{J}\boldsymbol{z}_{2}$ is the main term that guarantees the convergence of $\boldsymbol{z}_{2}$. Since the state trajectory is far away from the steady state after the sudden disturbance, the proportion term plays the major role in the control output short after the disturbance. However, when the system has converged a lot, the proportion term will not play a major role at that time. Therefore, the "convergence" part of the controller $-K_{2}\boldsymbol{M}_{2}\boldsymbol{J}\boldsymbol{z}_{2}$ and the "wrong direction repulsive force" part of the controller $- 2\boldsymbol{K}_{3}\boldsymbol{\Gamma}^{-1}\boldsymbol{D}_{\rho}\boldsymbol{\varepsilon}_{q}$ will effect against each other, and the error state variable will finally be "stuck" to a terminal state, which is not in the constraint region. Our explanation to this phenomenon is validated by a test: by changing the value of the controller gain $K_{\omega}$, we are able to change the terminal value of state variable. Further, the bigger the controller gain is, the smaller the terminal steady-state error is.
\end{remark}

The simulation result indicates that our proposed SAPPC scheme is able to handle the severe external disturbance properly. Even when the state trajectory is pulled out of the constraint region by the external disturbances, the controller is still able to lead the error state back to the convergence.

\subsection{Monte Carlo Simulation}
To evaluate whether the proposed SAPPC scheme is able to handle different conditions, a group of Monte-Carlo simulation is carried out to analyze the sensitivity of the proposed method. In the Monte-Carlo simulation, the initial three-axis euler angle of the spacecraft's attitude is considered to be a randomly chosen value in a given range as $\left[-85,85\right]^{\circ}$, which is able to cover most of the possible on-orbit attitude maneuver scenarios. In the Monte-Carlo simulation test, the simulation duration is $50s$ and 3000 times simulations are executed in the Monte Carlo simulation experiments.

For this Monte Carlo simulation campaign, the reference trajectory function (RPF) is designed as Table [\ref{tab:RPFpara2}].
\begin{table}[hbt!]
	\centering
	\begin{tabular}{|l|l|}
		\hline
		Initial Value of Exponential Function Part $\rho_{e0}$  &  $q_{evi}\left(0\right)$\\ \hline
		Terminal Value of Exponential Function Part $\rho_{e\infty}$ &  1e-6 \\ \hline 
		Coefficient of Exponential Function Part $l$ & 0.2 \\ \hline
		Convergence Time of the RPF   $t_{2}$              & 20 \\ \hline
		Terminal Value of the RPF   $g_{\infty}$         & 3e-5 \\ \hline
	\end{tabular}
	\caption{\label{tab:RPFpara2} Coefficients for the designed RPF in Monte Carlo Simulation test}
\end{table}
Note that the initial value of the designed RPF is decided by the initial condition of the state variable $q_{evi}\left(t\right)$. The system model, the attitude tracking task along with the external disturbance model are the same as those in subsection V.\ref{Simualtionpartone}.

\begin{figure}[hbt!]
	\centering  
	\includegraphics[scale = 0.5]{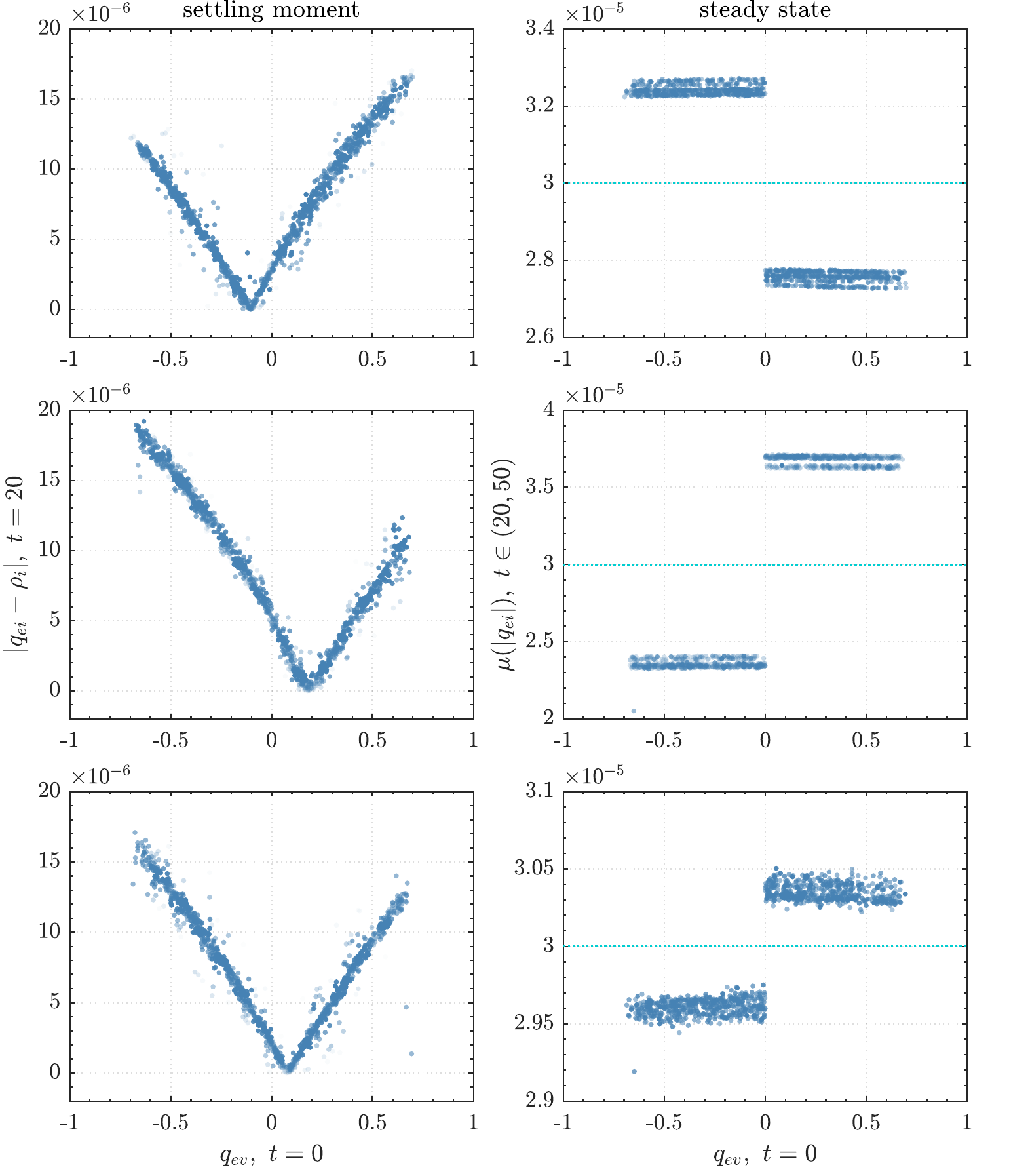}  
	\caption{Left: RPF Tracking Error At $t=20s$ $\quad$ Right: Steady-State Control Error}
	\label{fig_MC} 
\end{figure}

The simulation results of the Monte Carlo simulation are shown in Figure[\ref{fig_MC}][\ref{fig_MCST}][\ref{fig_MCST2}] as below, respectively. The left subplot Figure [\ref{fig_MC}] shows the deviation between the state trajectory $q_{evi}\left(t\right)$ and the RPF at $t=20s$, with each component is illustrated separately. The right subplot of the Figure [\ref{fig_MC}] shows the mean value of the state trajectory in $t\in\left[25,50\right]$. Figure [\ref{fig_MCST}][\ref{fig_MCST2}] shows the 3000-times state trajectory simulation result of the whole Monte Carlo simulation.

\begin{figure}[hbt!]
	\centering  
	\includegraphics[scale = 0.25]{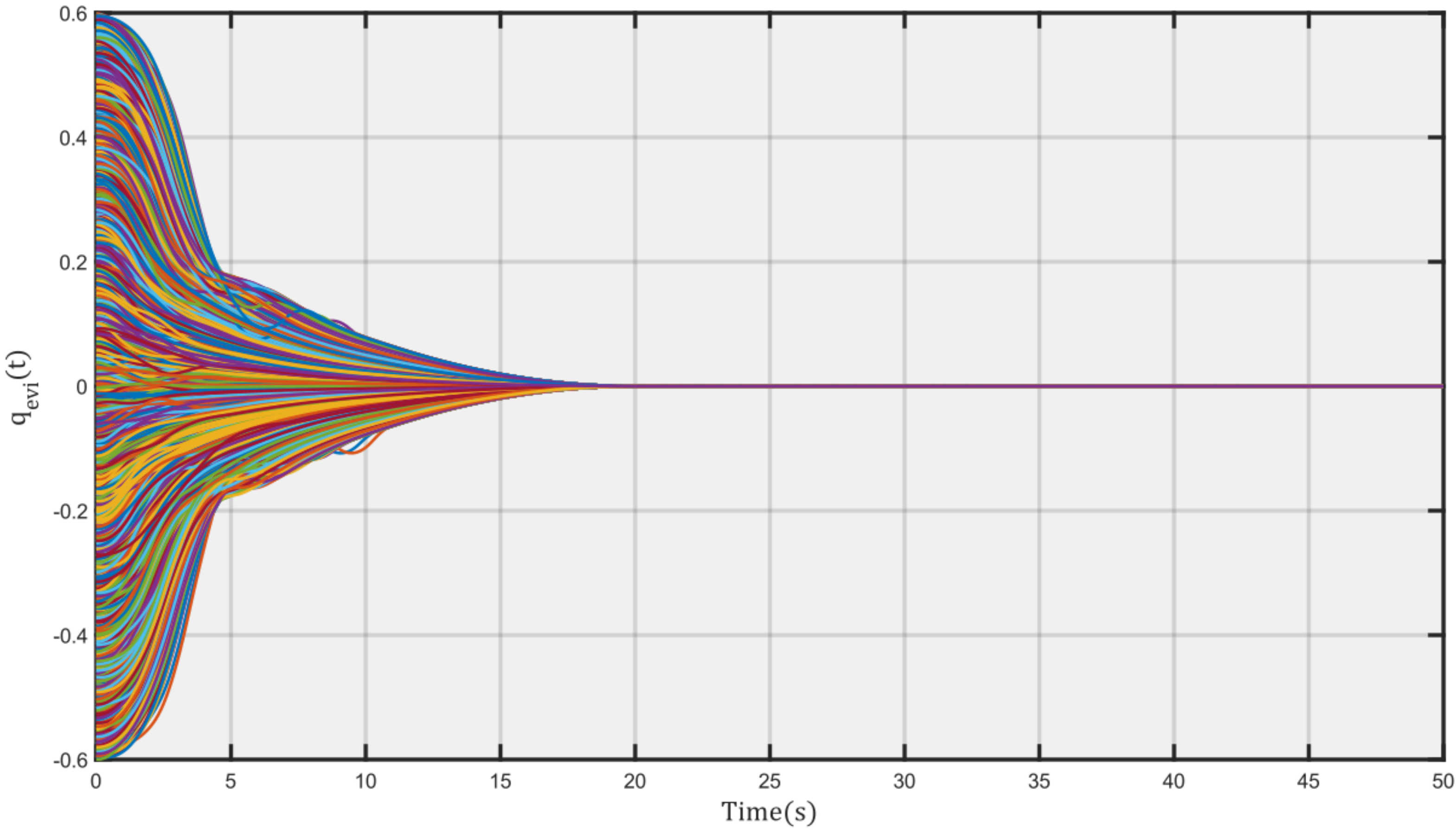}  
	\caption{Monte Carlo Simulation State Trajectory with $t_{2} = 20s$(3000 times Result)}
	\label{fig_MCST} 
\end{figure}

\begin{figure}[hbt!]
	\centering  
	\includegraphics[scale = 0.25]{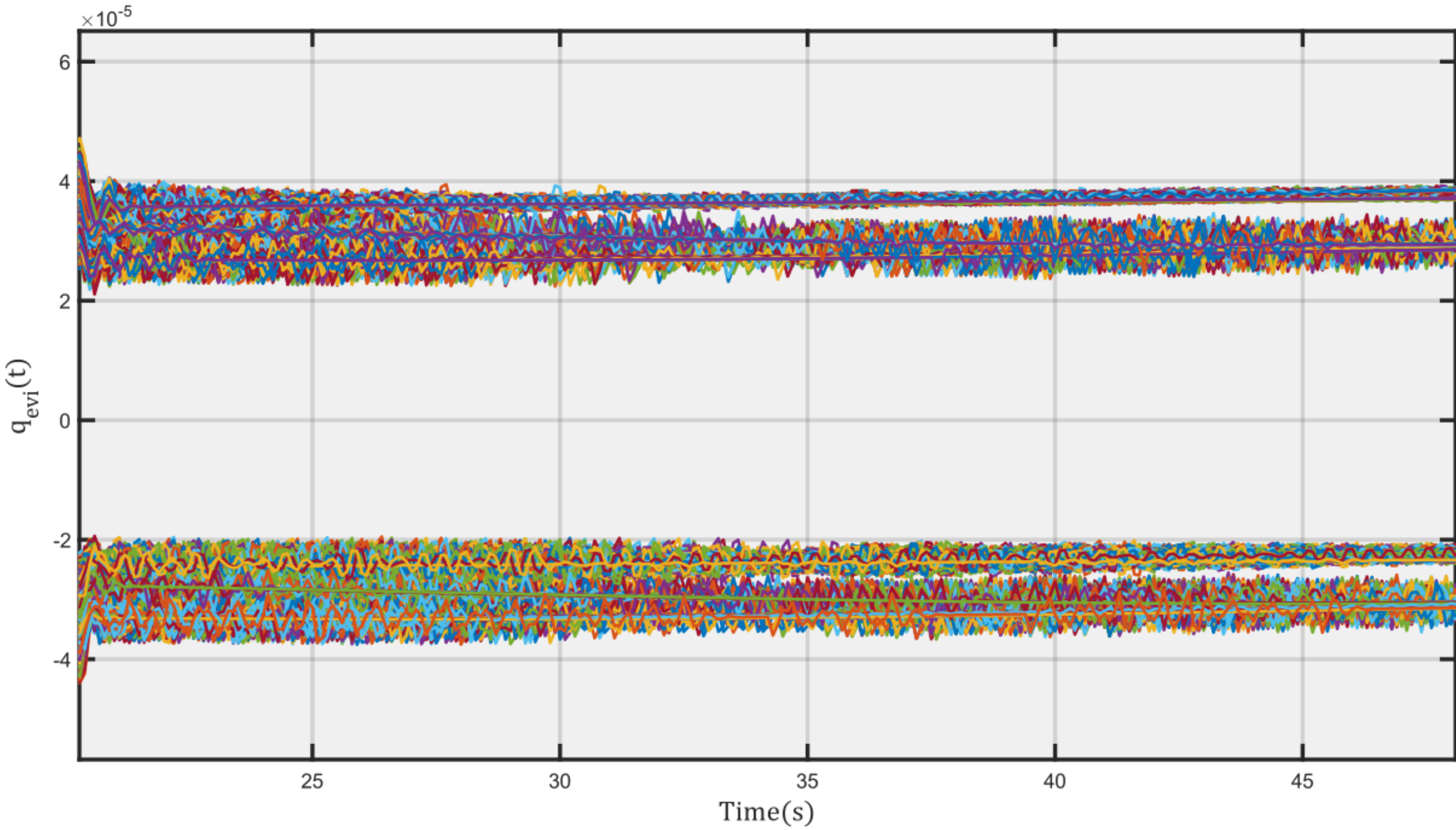}  
	\caption{Monte Carlo Simulation State Trajectory(Steady State) (3000 times Result)}
	\label{fig_MCST2} 
\end{figure}

We can observe that $\max|q_{evi} - \rho_{i}|\left(i=1,2,3\right)\left(t=20s\right)$ is smaller than 1e-4, which indicates that system has tracked the steady-state tightly at the settling moment. Further, it can be observe that for all the simulation cases, the steady-state control error can be limited to $5e-5$. From the Figure [\ref{fig_MCST}] we can observe that all the state trajectory is converged to the steady state at $t=20s$. For all the cases, it can be observed that all the state trajectory is converge to the neighborhood of the RPF at around $3e-5$, as shown in Figure [\ref{fig_MCST2}].

\section{Conclusion}
This paper has proposed a novel singularity-avoidance prescribed performance control scheme (SAPPC) to solve the attitude tracking control problem with preassigned controller performance requirements. To tackle the inherent singularity problem in the traditional PPC scheme, we first introduce the shear space mapping transformation to the error transformation procedure, providing a novel global non-singular error transformation procedure for the PPC scheme.
In order to alleviate the over-control problem of the traditional PPC schemes, we employ the time-varying state constraint boundary to provide appropriate constraint strengths at different control stages. Further, we design a smooth reference trajectory for the state responding to help achieve accurate transient performance requirements. 
The numerical simulation result indicates that these singularity problems can achieve the preassigned performance requirement even when there exists significant external disturbance, and the terminal steady-state control error is improved due to the improvement of the contradiction between the control effect and constraint's strength. Further, accurate control of the system responding is realized as the state trajectory could track the RPF tightly, and the state trajectory is able to converge to the steady state under the guidance of the given RPF.

Since the proposed SAPPC scheme is able to be theoretically combined with other control schemes, its application in other scenarios is worth further investigation. The proposed SAPPC scheme is able to design the system's state trajectory precisely, so there is potential value in its application in contemporary space missions, especially for collaborative attitude control problem.

\section{Appendix}
In this section, some significant inequalities and results are elaborated for the proof of the proposed controller, and the specific structure of the benchmark controller is also elaborated.
\subsection{Appendix A.}
\label{appendixA}
Consider the function as $f\left(x\right) = \exp\left(x^p\right)\quad p\in\left(0,1\right)$, take the time-derivative of $f\left(x\right)$, we can yield that:
\begin{equation}
	\dot{f}\left(x\right) = e^{x^{p}}px^{p-1} \ge 0
\end{equation}
It is obvious that for $p\in\left(0,1\right)$ and $x\in\left(0,+\infty\right)$, the function $f\left(x\right)$ is a monotonically increasing one.
Similarly, consider a function as $h\left(x\right) = x^{1-p} \quad p\in\left(0,1\right)$, take the time-derivative of $h\left(x\right)$, we have:
\begin{equation}
	\dot{h}\left(x\right) = \left(1-p\right)x^{-p} \ge 0
\end{equation}
Thus, the function $h\left(x\right)$ is a monotonically increasing function.
\subsection{Appendix B.  The Control Law of the Benchmark Controller 1 $\left(TraPPC\right)$}
\label{TraPPC}
For the benchmark controller $1$ named as "TraPPC", its PPC scheme has been elaborated in the existing literature \cite{wei_2021_overview}. The performance function of the TraPPC controller is the mostly applied exponential function type
Specially, $K=0.3$ is applied in this paper.
 Since the TraPPC benchmark controller has the same structure similar to the proposed SAPPC controller, the virtual control law, actual control law and the dynamic surface filter is the same as the proposed SAPPC controller.

\subsection{Appendix C.  The Control Law of the Benchmark Controller 2 $\left(BLFPPC\right)$}
\label{BLFPPC}
For the benchmark controller $2$ named as "BLFPPC" controller, its PPC scheme has been elaborated in the existing literature \cite{hu_adaptive_2018}. To realize the finite-time convergence, its performance function is replaced by the FTPPF applied in the literature as \cite{gao_finite-time_2021}. The expression of the employed FTPPF is expressed as follows:
\begin{equation}
	\begin{aligned}
		\begin{split}	\rho\left(t\right) = 	
			\left \{	
			\begin{array}{ll}
				\left(\rho^{m}_{0} - m\lambda\cdot t\right)^{1/m} + \rho_{T_{f}} & t\in\left(0,T_{f}\right) \\
				\rho_{T_{f}}   & t\in\left[T_{f},+\infty\right)
			\end{array}				
			\right.	
		\end{split}
	\end{aligned}
\end{equation}
where $\rho_{0}$, $m$, $\lambda$, $\rho_{T_{f}}$ are the coefficients that needs designing. $\rho_{T_{f}}$ denotes that terminal value of the performance function, $\rho_{0}$ is decided by the initial value of the FTPPF as $\rho_{0} = \rho\left(0\right) - \rho_{T_{f}}$, $m\cdot\lambda = \frac{\rho^{m}_{0}}{T_{f}}$. $m$ is a positive value that satisfies $m = a_{1}/a_{2} \in\left(0,1\right)$, where $a_{1}$, $a_{2}$ are a positive odd integer and a positive even integer respectively, $T_{f}$ is the expected settling time. This FTPPF satisfies $\lim_{t\to T_{f}}\rho\left(t\right) = \rho_{T_{f}}$, also, differentiable on $t\in\left(0,+\infty\right)$.
Note that the performance function of this scheme can be indicated individually as $\rho_{li}$, $\rho_{ui}$. Thus, we set the introduced FTPPF as $\rho_{ui}$ and let $\rho_{li} = 0$ to avoid overshooting.
The control system is designed as follows:
\begin{equation}
	\begin{aligned}
		\boldsymbol{\alpha} &= \boldsymbol{\Gamma}^{-1}
		\left[-\boldsymbol{K}_{1}\boldsymbol{q}_{ev} + \frac{1}{2}\boldsymbol{K}_{1}\boldsymbol{P}_{lu} + \boldsymbol{S}_{v}\right] \\
		\boldsymbol{u} &= -\boldsymbol{K}_{2}\boldsymbol{J}\boldsymbol{z}_{2} - 2\boldsymbol{K}_{3}\boldsymbol{\Gamma}^{-1}\boldsymbol{D}_{\rho}\boldsymbol{\varepsilon}_{q} -\boldsymbol{W}_{0} + \boldsymbol{J}\dot{\boldsymbol{S}_{d}} \\
		&\quad - 
		D_{m}\text{vec}\left(\tanh\left(\frac{z_{2i}}{\mu_{i}}\right)\right)  \\
	\end{aligned}
\end{equation}
where $\boldsymbol{K}_{1}\in\mathbb{R}^{3\times 3}$, $\boldsymbol{K}_{2}\in\mathbb{R}^{3\times 3}$ and $\boldsymbol{K}_{3}\in\mathbb{R}^{3\times 3}$ represents the diagonal matrix whose diagonal line is consisted by the controller parameter $K_{1}$,$K_{2}$,$K_{3}$ respectively.
$\boldsymbol{P}_{lu}$, $\boldsymbol{D}_{\rho}$ and $\boldsymbol{S}_{v}$ are defined as follows:
\begin{equation}
	\begin{aligned}
		\boldsymbol{P}_{lu} &= \boldsymbol{\rho}_{l} + \boldsymbol{\rho}_{u} \\ 
		\boldsymbol{S}_{b} &= \text{vec}\left[\dot{\rho}_{li}\rho_{ui} - \dot{\rho}_{ui}\rho_{li}\right]\left(i=1,2,3\right) \\
		\boldsymbol{S}_{v} &= \text{diag}(1/(\rho_{ui} - \rho_{li}))\cdot
		\left[\text{diag}\left(\dot{\rho}_{ui} - \dot{\rho}_{li}\right)\boldsymbol{q}_{ev} + \boldsymbol{S}_{b}\right]\\
		\boldsymbol{D}_{\rho} &=\text{diag}\left(\frac{1}{\left(1-\varepsilon_{qi}^{2}\right)\left(\rho_{ui} - \rho_{li}\right)}\right)
	\end{aligned}
\end{equation}
It should be noticed that $-\boldsymbol{W}_{0}$, $\boldsymbol{J}\dot{\boldsymbol{S}}_{d}$ and $D_{m}$ are known.

\bibliographystyle{IEEEtran}
\bibliography{IEEEabrv,ast_reference_test}

%\biographysection

\begin{IEEEbiography}
	[{\includegraphics[width=1in,height=1.25in,clip,keepaspectratio]{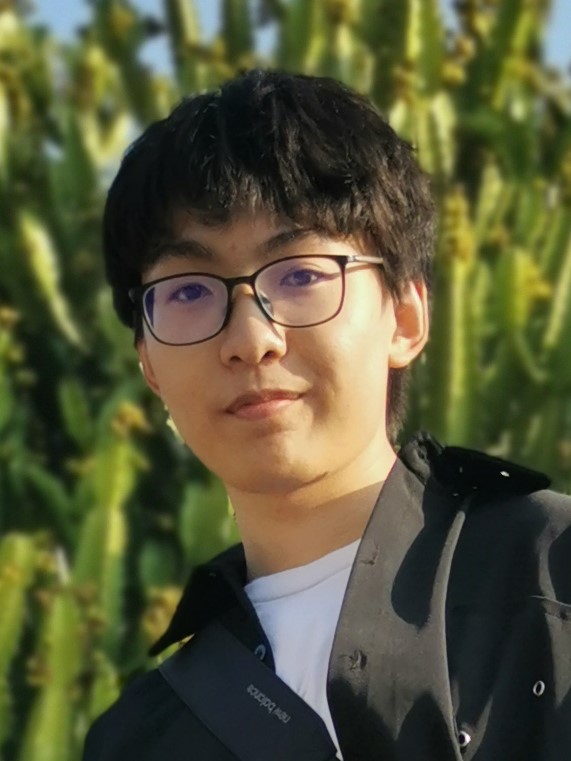}}]
	{Jiakun Lei}{\space} 
received the B.S. degree from Dep.Automation, University of Electronic Science and Technology of China, Chengdu, China, in 2019. He is currently working toward the Ph.D. degree in aeronautical and astronautical science and technology in Zhejiang University, Hangzhou, China.
His research interests include attitude planning and control under constraints, attitude control of spacecraft with complex structure and attitude formation.
\end{IEEEbiography}

\begin{IEEEbiography}
	[{\includegraphics[width=1in,height=1.25in,clip,keepaspectratio]{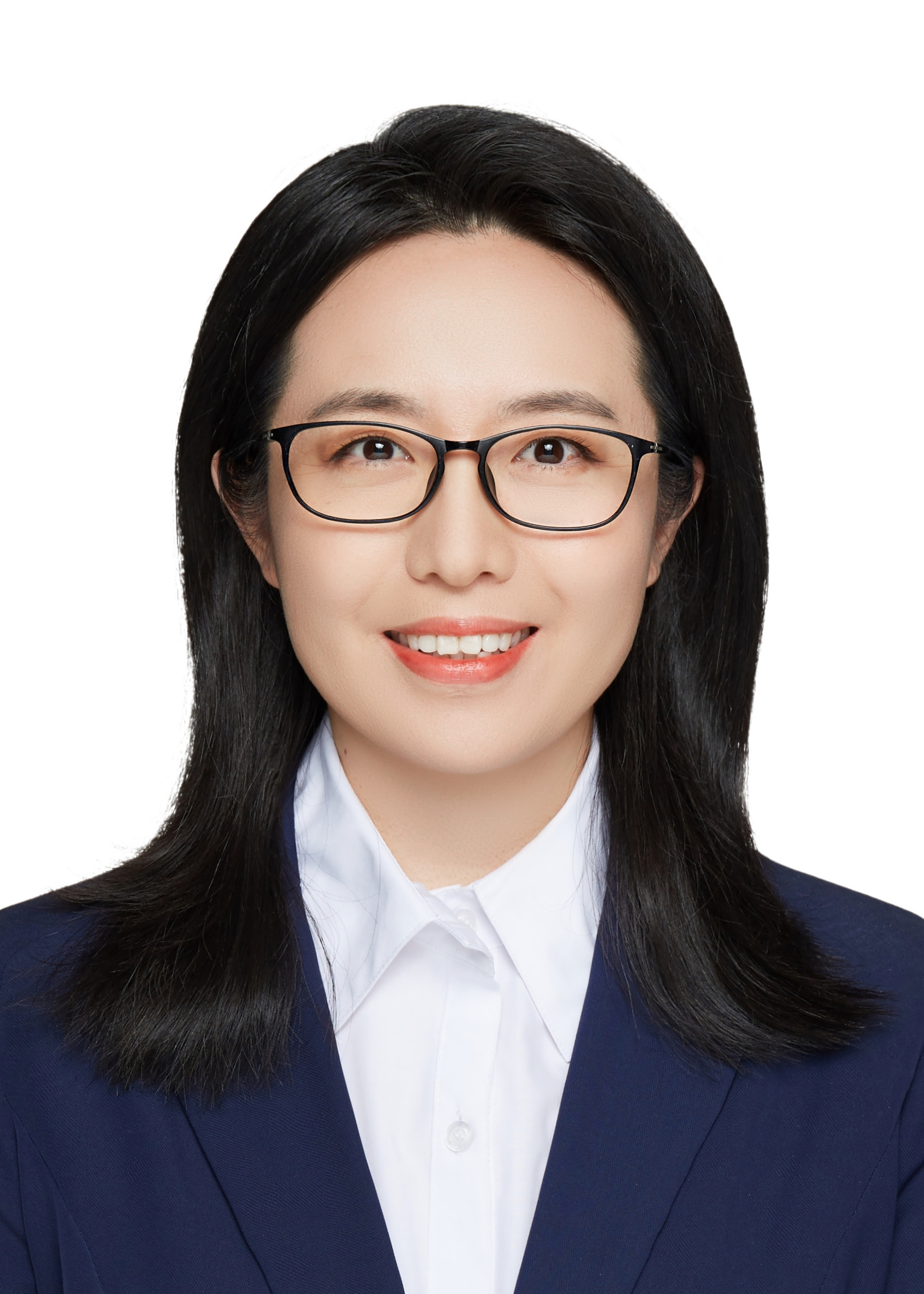}}]
	{Tao Meng}{\space}
	 received the B.S. degree from Electronic science and technology, Zhejiang university, Hangzhou, China, in 2004, the M.S. degree from Electronic science and technology, Zhejiang university, Hangzhou, China, in 2006, and the Ph.D. degree from Electronic science and technology, Zhejiang university, Hangzhou, China, in 2009. She is currently the Professor of the School of Aeronautics and astronautics. Her research interest include attitude control, orbit control and constellation formation control of micro-satellite.
\end{IEEEbiography}

\begin{IEEEbiography}
	[{\includegraphics[width=1in,height=1.25in,clip,keepaspectratio]{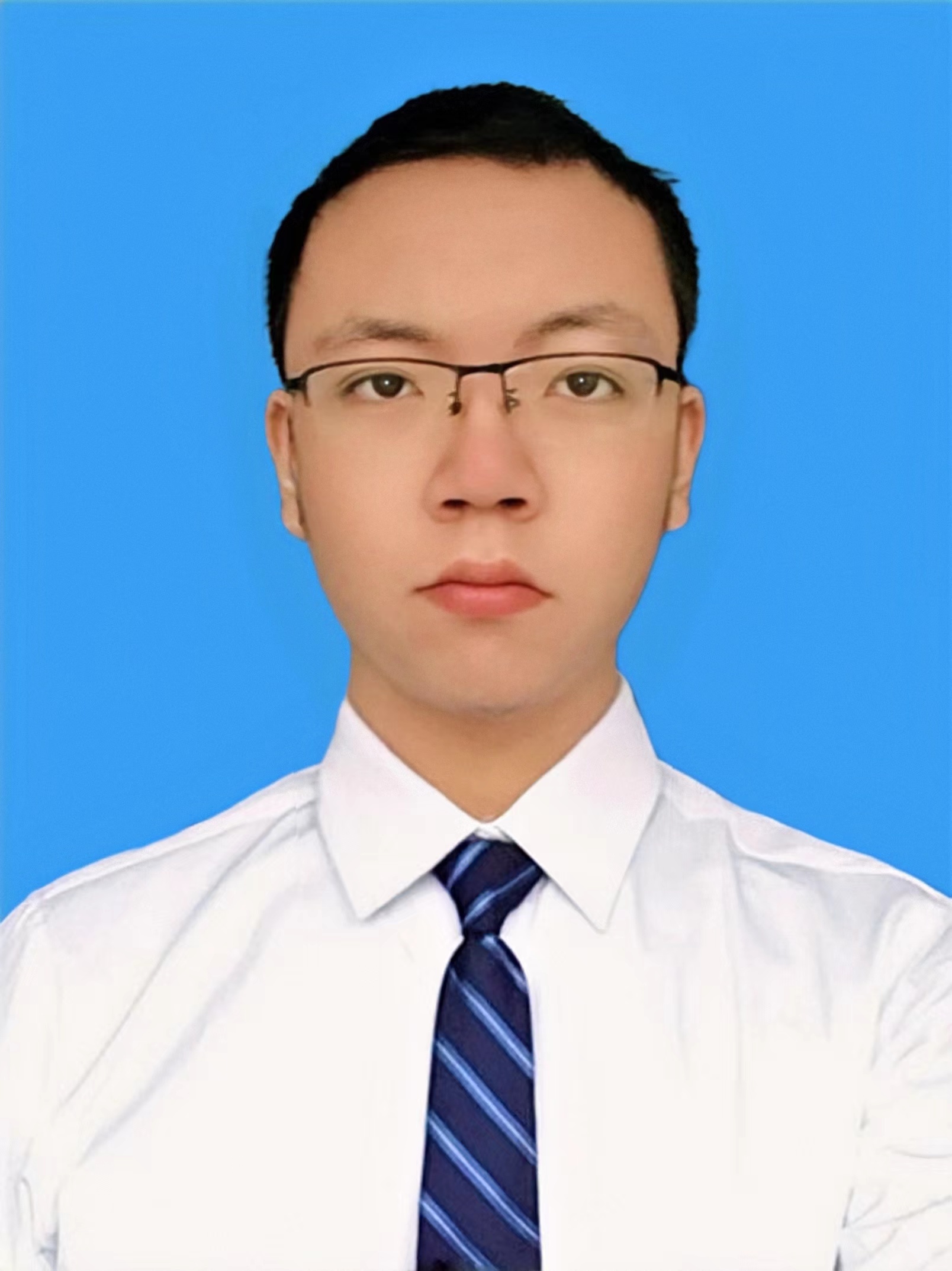}}]
	{Weijia Wang}{\space}
    received the B.S. degree from aerospace engineering, University of Electronic Science and Technology of China, Chengdu, China, in 2020. He is currently working toward the Ph.D. degree in aeronautical and astronautical science and technology in Zhejiang University, Hangzhou, China.	His research interests include model predictive control and reinforcement-learning-based control for 6-DOF spacecraft formation. 
\end{IEEEbiography}

\begin{IEEEbiography}
	[{\includegraphics[width=1in,height=1.25in,clip,keepaspectratio]{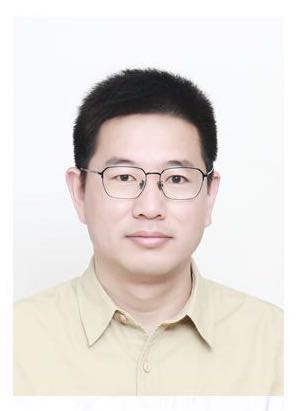}}]
	{Heng Li}{\space}
	received the B.S. degree from Mechanical Engineering, Hebei University of Technology, Tianjin, China, in 2010, and the M.S. degree from agricultural robotics technology, China Agricultural University, Beijing, China, in 2013. He is currently an engineer at Micro-Satellite Research Center of Zhejiang University, Hangzhou, China. His main research interest are software performance optimization of attitude and orbit control system.
\end{IEEEbiography}

\begin{IEEEbiography}
	[{\includegraphics[width=1in,height=1.25in,clip,keepaspectratio]{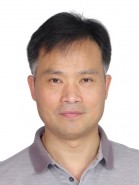}}]
	{Zhonghe Jin}{\space}
    received the Ph.D. degree from major of microelectronics and solid electronics, Zhejiang University, Hangzhou, China, in 1998. He is now the deputy director of Zhejiang University and the director of the center for micro satellite research. His main research areas include micro-satellites, MEMS etc.
	
\end{IEEEbiography}

\end{document}